\documentclass{article}
\usepackage[T1]{fontenc}
\usepackage[latin9]{inputenc}
\usepackage{amsmath}
\usepackage{amssymb}
\usepackage{mathtools}
\usepackage{mathtools}
\usepackage{amsthm}
\usepackage{graphicx}
\usepackage{pgfplots}
\pgfplotsset{compat=1.18}
\usepackage{subcaption}
\usepackage{amsmath}
\usepackage{appendix}
\usepackage[colorlinks=true,allcolors=black]{hyperref}
\usepackage{cleveref}
\usepackage{xcolor}
\usepackage{comment}
\usepackage{thm-restate}
\usepackage{natbib}
\usepackage{fullpage}
\usepackage{array} 
\usepackage{tikz}
\usepackage{caption}
\usepackage{enumitem}

\DeclareMathOperator*{\argmax}{arg\,max}

\usepackage[symbol]{footmisc}
\usepackage{placeins}
\usepackage{nicefrac,bbm}

\allowdisplaybreaks

\newtheorem{theorem}{Theorem}[section]
\newtheorem*{theorem*}{Theorem}
\newtheorem{lemma}[theorem]{Lemma}
\newtheorem{proposition}[theorem]{Proposition}

\newtheorem*{corollary*}{Corollary}

\newtheorem{remark}[theorem]{Remark}

\newcommand{\cI}{\mathcal{I}}

\newcommand{\R}{\mathbb{R}}

\newcommand{\SW}{\textup{\textsc{Sw}}}
\newcommand{\OPT}{\textup{\textsc{Opt}}}
\newcommand{\OPTLL}{\textup{\textsc{OptWithoutLL}}}

\newcommand{\EQ}{\textup{\textsc{Ua}}}

\newcommand{\ID}{\textup{\textsc{Id}}}
\newcommand{\DD}{\textup{\textsc{Dd}}}

\newcommand{\E}[1]{\mathbb{E}\left[#1\right]}

\newcommand{\bQ}{\textbf{Q}}

\newcommand{\bc}{\textbf{c}}

\newcommand{\bq}{\textbf{q}}

\newcommand{\bu}{\textbf{u}}

\newcommand{\bw}{\textup{\textbf{w}}}

\newcommand{\bzero}{\textbf{0}}
\newcommand{\bone}{\textbf{1}}

\makeatletter
\newcommand{\vast}{\bBigg@{4}}
\newcommand{\Vast}{\bBigg@{5}}
\makeatother

\def\EqStyle{}
\let\citet\cite
\def\QLEXPR{\frac{q_{\ell}}{\sum_{i \in [\ell]} q_i}}

\usepackage[]{color-edits}%
\addauthor{PD}{black}

\newcommand\blfootnote[1]{
  \begingroup
  \renewcommand\thefootnote{}
  \NoHyper\footnote{#1}\endNoHyper
  \addtocounter{footnote}{-1}
  \endgroup
}

\title{Anonymous Contracts}

\date{}

\author{
Johannes Br\"ustle \\ Sapienza University of Rome
\and
Paul D\"utting \\ Google Research
\and
Stefano Leonardi \\ Sapienza University of Rome
\and
Tomasz Ponitka \\ Tel Aviv University
\and
Matteo Russo \\ EPFL
}

\begin{document}

\maketitle

\blfootnote{\hspace*{-2.5em}
S. Leonardi is partially supported by the ERC Advanced Grant 788893 AMDROMA ``Algorithmic and Mechanism Design Research in Online Markets'', by the FAIR (Future Artificial Intelligence Research) project PE0000013, funded by the NextGenerationEU program within the PNRR-PE-AI scheme (M4C2, investment 1.3, line on Artificial Intelligence), and by the PNRR MUR project IR0000013-SoBigData.it, and by the MUR PRIN grant 2022EKNE5K (Learning in Markets and Society). T. Ponitka was supported by ERC grant 101170373, an Amazon Research Award, NSF-BSF grant 2020788, Israel Science Foundation grant 2600/24, and a TAU Center for AI and Data Science grant. A portion of this work was completed while M. Russo was a student at Sapienza University of Rome.}

\begin{abstract}
We study a multi-agent contracting problem where agents exert costly effort to achieve individually observable binary outcomes. While the principal can theoretically extract the full social welfare using a discriminatory contract that tailors payments to individual costs, such contracts may be perceived as unfair. In this work, we introduce and analyze \emph{anonymous contracts}, where payments depend solely on the total number of successes, ensuring identical treatment of agents.

We first establish that every anonymous contract admits a pure Nash equilibrium. However, because general anonymous contracts can suffer from multiple equilibria with unbounded gaps in principal utility, we identify \emph{uniform anonymous contracts} as a desirable subclass. We prove that uniform anonymous contracts guarantee a unique equilibrium, thereby providing robust performance guarantees.

In terms of efficiency, we prove that under limited liability, anonymous contracts cannot generally approximate the social welfare better than a factor logarithmic in the spread of agent success probabilities. We show that uniform contracts are sufficient to match this theoretical limit. Finally, we demonstrate that removing limited liability significantly boosts performance: anonymous contracts generally achieve an $O(\log n)$ approximation to the social welfare and, surprisingly, can extract the full welfare whenever agents' success probabilities are distinct. This reveals a structural reversal: widely spread probabilities are the hardest case under limited liability, whereas identical probabilities become the hardest case when limited liability is removed.
\end{abstract}

\newpage
\clearpage

\section{Introduction}\label{sec:intro}

Contract design is a pillar of microeconomic theory \cite{Mirrlees99,Ross73,GrossmanHart83,Holmstrom79,NobelPrize16}. In this work, we focus on a multi-agent contracting problem with \emph{individual} outcomes 
\cite{CastiglioniM023}. In this setting, each of $n$ agents has a binary choice of exerting effort or not. The cost of exerting effort is $c_i$ for agent $i$. If agent $i$ exerts effort, they succeed in their individual task with probability $q_i$. We focus on the case where the principal's expected reward is additive across agents, such that the (normalized) expected reward from a set of agents $S$ exerting effort is $f(S) = \sum_{i \in S} q_i$.

As observed in \cite{CastiglioniM023}, optimal contracts in this specific setting are remarkably simple yet highly discriminatory. The principal can extract the entire social welfare as surplus by incentivizing precisely those agents $i \in [n]$ for which $q_i > c_i$, paying them just enough to cover their costs. Specifically, the optimal payment to agent $i$ is $c_i/q_i$ for success and zero otherwise. While efficient, these contracts may be perceived as unfair: two agents that succeed and contribute equally to the principal's utility can receive vastly different payments. For instance, if $q_1 = q_2 = 1/2$ but $c_1 = 0.1$ and $c_2 = 0.01$, agent $1$ is paid $0.2$ while agent $2$ is paid only $0.02$.

Motivated by this discrepancy, we propose and analyze \emph{anonymous contracts}. These contracts avoid such imbalance by constraining the principal to offer a common payment scheme to all agents that succeed, regardless of their individual costs $c_i$. 
In what follows, we will use $\bw = (w_1, \ldots, w_n)$ to denote the payments of such a contract, with $w_j$ designating the payment 
to each successful agent
when exactly $j$ agents succeed. 
By design, such contracts are inherently less discriminatory and may be more appealing in many real-world settings.

Observe that payment vector $\bw$ is unrestricted and can have both positive and negative components. When all of its components $w_j \ge 0$, we say that the principal has to satisfy \emph{limited liability}, which is the main focus of this work, and otherwise the setting is \emph{without limited liability}.

\subsection{Our Contributions and Techniques}

Our goal in this work is to quantify the principal's loss from using an anonymous contract, rather than an unrestricted contract.

\paragraph{Structural Insights.}

Our first objective, addressed in Section~\ref{sec:pne}, is, to establish whether anonymous contracts are well-defined for arbitrary anonymous payments. We show that this is indeed the case in a strong sense: for any payment vector $\bw \in \mathbb{R}^n$ (i.e., with possibly negative payments), there exists a subset $S \subseteq [n]$ such that exactly the agents in $S$ are incentivized under payment $\bw$. In particular, no agent in $S$ benefits from deviating by not exerting effort, and no agent outside $S$ would prefer to exert effort given the expected payment. In other words, every payment vector $\bw$ admits a \emph{pure Nash equilibrium}:

\medskip
\noindent \textbf{Theorem (Existence of PNEs, \Cref{thm:equilibrium}).} For all payment vectors $\bw \in \R^n$, the anonymous contract defined by $\bw$ admits a (pure) Nash equilibrium.
\medskip

Our proof leverages the symmetry in reward functions, showing that if no equilibrium existed, the iterative addition or removal of agents would lead to a contradiction due to reward cancellations. By carefully tracking how expected payoffs shift with deviations, we establish that some subset of agents always reaches a stable equilibrium. 

Beyond existence, we construct explicit examples demonstrating that equilibria need not be unique. For a fixed $\bw$, multiple stable agent subsets can emerge, indicating that identical contracts may yield vastly different outcomes depending on agent responses. This highlights the importance of equilibrium selection in anonymous contracts:

\medskip
\noindent \textbf{Proposition (Unbounded Gap between PNEs, \Cref{prop:unbounded}).} 
There exists an instance and a payment vector $\bw \in \R^n_+$
such that the gap in principal utility between the best and the worst equilibrium induced by $\bw$ is unbounded.
\medskip

The conclusion of \Cref{prop:unbounded} poses a critical risk to the principal: since the gap between best and worst equilibria can be arbitrarily large, she cannot safely rely on a contract that performs well only in a favorable equilibrium. This ambiguity renders best-case guarantees fragile. Consequently, we seek contracts admitting a \emph{unique} equilibrium
In particular, we define \emph{uniform anonymous contracts} as anonymous contracts that pay the same amount irrespective of the number of successes, i.e., set $w_1 = \ldots = w_n = w$.

\medskip

\noindent \textbf{Proposition (Unique PNE for Uniform Anonymous Contracts, \Cref{prop:equal}).} A uniform anonymous contract $\bw$ admits a unique (pure) Nash equilibrium. Moreover, this (unique) equilibrium is a \emph{prefix} set, that is a set $S$ consisting of 
$|S|$
agents sorted in nondecreasing order of $\tfrac{c_i}{q_i}$.
\medskip

\paragraph{Results With Limited Liability.} Next, we analyze the efficiency of anonymous contracts from the principal's perspective in the limited liability case. The first insight is a strong impossibility result that relates the utility of the principal under any anonymous contract to the welfare. We 
parameterize instances by
the largest gap between highest and smallest agent's success probability $Q = (\max_i q_i)/(\min_i q_i)$:

\medskip
\noindent \textbf{Theorem (Impossibility under Limited Liability, \Cref{thm:negativegeneral}).} 
For any $n$ and $Q$,
there exists an instance such that the principal's utility collected by an optimal anonymous contract is 
$\Omega(\min \{ n, \log Q \})$
times smaller than the social welfare of the same instance.
\medskip

Observe that, in the worst case ($Q = 2^n$), the performance of any anonymous contract is capped at a linear approximation of the social welfare. In light of this strong impossibility result, we highlight the second crucial property of \emph{uniform anonymous contracts}, i.e., that they match this rate almost exactly (up to lower order terms):

\medskip
\noindent \textbf{Theorem (Approximation via Uniform Anonymous Contracts, \Cref{thm:Qbounds}).} The principal's utility collected by an optimal uniform anonymous contract is always at most $O(\min\{n, \log(Qn)\})$ times smaller than the social welfare.
\medskip

Our key insight to prove this result is that it is always profitable to incentivize a prefix of agents, i.e., the agents indexed from 1 to some \( k \) according to the ratio $c_i/q_i$. By optimizing over such prefixes, we establish a structural link between the optimal uniform contract and social welfare, leading to our approximation bound. In addition to the above two desirable properties, optimal uniform anonymous contracts are also efficiently computable, in \( \text{poly}(n) \) time.

\paragraph{Results Without Limited Liability.} In the last part of this work, we circumvent the impossibility bounds of \Cref{sec:general_instances} by lifting the limited liability assumption. We demonstrate that allowing for negative payments dramatically improves the power of anonymous contracts. In the general setting, this flexibility yields a logarithmic approximation to the social welfare. In addition, under the mild but \emph{necessary} condition that agents have distinct success probabilities, this approach enables the principal to extract the full social welfare. 

\medskip
\noindent \textbf{Theorem (Anonymous Contracts without Limited Liability, \Cref{thm:positive_without_ll} and \Cref{thm:invertible-noLL}).} There always exists a payment vector achieving an $O(\log n)$-approximation to the social welfare. This is (asymptotically) tight, even among anonymous contracts that satisfy limited liability, when the agents' success probabilities can all be the same. In contrast, if all agents' success probabilities $q_i$ are distinct, there exists a payment vector that extracts the full social welfare.
\medskip

This result highlights a surprising reversal regarding the structural hardness of instances. Under limited liability, the worst-case scenarios for anonymous contracts arise when the agents' success probabilities $q_i$'s are exponentially spread. On the contrary, removing limited liability flips this dynamic: as distinct probabilities facilitate full welfare extraction, the hardest instances to approximate become those where all agents have identical success probabilities (see \Cref{fig:ll_vs_noll}).

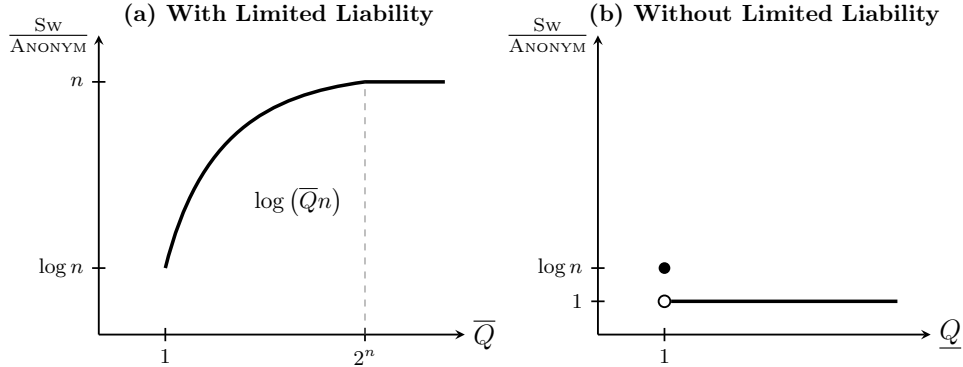
\begin{figure}[htbp]
\centering
\resizebox{0.8\linewidth}{!}{
\tikzset{every picture/.style={line width=0.75pt}}

\begin{tikzpicture}[
    axis/.style={->, >=stealth, thick},
    label/.style={font=\large},
    tick/.style={font=\small},
    curve/.style={thick, line width=1.5pt},
    dashed_line/.style={dashed, thin, gray}
]

    \begin{scope}[local bounding box=leftgraph]
        \node[font=\bfseries] at (2.7, 4.8) {(a) With Limited Liability};

        \draw[axis] (0,0) -- (5.5,0) node[right, label] {$\overline{Q}$};
        \draw[axis] (0,0) -- (0,4.5) node[left, label] {$\tfrac{\SW}{\textsc{Anonym}}$};

        \draw (0.1, 1) -- (-0.1, 1) node[left, tick] {$\log n$};
        \draw (0.1, 3.8) -- (-0.1, 3.8) node[left, tick] {$n$};

        \draw (1, 0.1) -- (1, -0.1) node[below, tick] {$1$};
        \draw (4, 0.1) -- (4, -0.1) node[below, tick] {$2^n$};

        \draw[curve] (1, 1) .. controls (1.5, 3) and (2.5, 3.6) .. (4, 3.8) -- (5.2, 3.8);

        \draw[dashed_line] (4, 0) -- (4, 3.8);

        \node at (3, 2) {$\log\left(\overline{Q}n\right)$};
    \end{scope}

    \begin{scope}[xshift=7.5cm, local bounding box=rightgraph]
        \node[font=\bfseries] at (2.5, 4.8) {(b) Without Limited Liability};

        \draw[axis] (0,0) -- (5,0) node[right, label] {$\underline{Q}$};
        \draw[axis] (0,0) -- (0,4.5) node[left, label] {$\tfrac{\SW}{\textsc{Anonym}}$};

        \draw (0.1, 0.5) -- (-0.1, 0.5) node[left, tick] {$1$};
        \draw (0.1, 1) -- (-0.1, 1) node[left, tick] {$\log n$};

        \draw (1, 0.1) -- (1, -0.1) node[below, tick] {$1$};

        \fill (1, 1) circle (2.5pt);

        \draw[curve] (1.1, 0.5) -- (4.5, 0.5);

        \draw[thick, fill=white] (1, 0.5) circle (2.5pt);

    \end{scope}

\end{tikzpicture}
}
\caption{(a) Under limited liability, the worst-case gap between social welfare $\SW$ and anonymous contracts $\textsc{Anonym}$ grows as $\min\{\log(\overline{Q} n), n\}$, where $\overline{Q} = \max_{i,j} {q_i}/{q_j}$. Without limited liability, if $\underline{Q} = \min_{i \neq j} {q_i}/{q_j} = 1$, the worst-case gap stays at $\log n$; once $\underline{Q} > 1$, this gap becomes exactly $1$.}
\label{fig:ll_vs_noll}
\end{figure}

\subsection{Related Work}\label{sec:rel-work}

Our work is situated within the burgeoning field of algorithmic contract theory (see, e.g.,  \cite{BabaioffFN06, BabaioffFNW12} and \cite{dutting2019simple}). We refer the reader to \cite{DuttingFT24} and \cite{Feldman26} for comprehensive surveys.

\paragraph{Multi-Agent Contracts with Individual Outcomes.}
We adopt a multi-agent contracting model with individually observable outcomes. The  most relevant predecessors in this direction are \cite{CastiglioniM023,GoelHC24}.

\citet{CastiglioniM023} 
consider a multi-agent setting, where each agent takes one of $\ell$ actions, and for each agent and action pair, there is a distribution over $m$ possible outcomes. The joint outcome is a tuple of outcomes, one from each agent. They observe that, while the principal could in principle make each agent's payment contingent on the entire vector of outcomes, an optimal contract makes payments to each agent, that are contingent only on the agent's individual outcome not that of other agents. 
An important consequence of this is that it  
completely decouples the problem of incentivizing a certain profile of actions as cheaply as possible across the agents. However, agents are coupled through the fact that the principal enjoys a reward for each tuple of outcomes, and this is not necessarily separable across agents. Indeed, the focus of their work is on IR-submodular and DR-supermodular rewards. Our problem is quite different because, although anonymous contracts present a class of simple contracts, they introduce an interdepence between agents that necessitates a careful equilibrium analysis. Importantly, the problem of designing optimal anonymous contracts is mathematically interesting and challenging already in the natural case where the principal's reward is additive across agents.

Recent work on Luce contracts \cite{GoelHC24} 
examines contract design where 
agents select success probabilities \( p_i \), incurring convex costs \( c_i(p_i) \), while a principal maximizes a weighted sum of successes under a budget constraint. Contracts allocate budget fractions based on which agents succeeded (resp.~failed).
In equilibrium, optimal contracts distribute the entire budget among successful agents, with agents that fail receiving nothing. A key challenge is the exponential search space over possible contracts. To address this, the authors introduce Luce contracts, a restricted class of contracts that imposes agent-specific orderings and weights, and yet characterizes the structure of optimal contracts. However, Luce contracts remain inherently non-anonymous, and it is unclear whether finding the right agent-specific ordering and weights is a task that can be performed in polynomial time.

In summary, while we work in a related setting as \cite{CastiglioniM023,GoelHC24} insisting on anonymity leads to a very different optimization problem, with incomparable challenges.

\paragraph{Multi-Agent Contracts with Joint Outcomes.}
While we focus on a model with individually observable outcomes, there is a substantial body of work %
on multi-agent contracts with a \emph{jointly observable} outcome \cite[e.g.,][]{BabaioffFN06,BabaioffFNW12,DuettingEFK23}. The fact that the principal can only observe a joint outcome (e.g., whether a project succeeded or not) constrains the principal's ability to fine-tune payments, and makes results for this model largely incomparable to results in the setting with individual outcomes.

This line of work studies the computational aspects of the model under binary actions \cite{BabaioffFN06,BabaioffFNW12,DuettingEFK23}, combinatorial actions \cite{DuettingEFK25}, the presence of budget constraints \cite{AharoniHT25,FeldmanGPS25,FeldmanGPS26a}, as well as various notions of equilibria \cite{dutting2025black} among others.
Importantly, in close relation to our study, several recent studies 
incorporate fairness considerations. 
\cite{FeldmanGPS26b,DingLS26} study 
equal-pay contracts, in which the principal may hire a subset of agents but must pay all hired agents the same amount, and show that, in the worst case, the gap between equal-pay contracts and optimal contracts is $\Theta(\log n / \log \log n)$, where $n$ is the number of agents. \cite{CastiglioniCL25b} study contracts that are envy-free under swaps, meaning that no two agents would obtain higher utility under the  new equilibrium obtained by exchanging their contracts. Notably, equal-pay contracts are envy-free; however, they are not anonymous, since the principal may select which agents to incentivize and is only constrained to use identical payments for those agents.

Finally, \cite{AlonCCELT25} introduce a model in which the principal has a set of available projects and assigns each agent to at most one project. Envy-free contracts in this setting are studied by \cite{CastiglioniCL25a}.

\paragraph{Contracts for Typed Agents.}
Another line of work studies contract design when the agent has an unknown \emph{type}, and the principal is only given a distribution over types \cite{AlonDT21,AlonDLT23,CastiglioniCLXZ25,CastiglioniMG22,GuruganeshSW023,GuruganeshSW21,CastiglioniMG23,FengMaXiao26}. When the distribution over types is itself unknown, a further line of work considers the resulting learning problem \cite{ZhuBYWJJ23,CohenDK23,HoSV16,BacchiocchiCGM25,DuttingFPS25,ChenCDH24,Hogsgaard26}. This perspective is closely related to our model of anonymous contracts. In particular, our analysis in \Cref{sec:uniqueness_of_pne_under_uniform_anonymous_contracts} implies that uniform anonymous contracts for an instance with $n$ agents in our model are mathematically equivalent to single contracts for an instance with a uniform distribution over $n$ agent types, where the principal's utility and social welfare are scaled by a factor of $1/n$. Under this correspondence, results on approximation guarantees for single contracts in the Bayesian setting translate directly to bounds for uniform anonymous contracts in our model. 
In particular,
\cite[Theorem 3]{GuruganeshSW21} shows that in their model with binary actions and binary outcomes, a single contract achieves a $\Theta(\log n)$-approximation to social welfare when either all agents have identical costs or identical success probabilities, which yields the same guarantee in our setting; see \Cref{sec:equal_prob,sec:equal_cost}, where we include proofs in our model for completeness. 
We emphasize that this correspondence between our model and typed-agent models does not extend to non-uniform anonymous contracts, which lack a natural counterpart in the Bayesian framework.

\paragraph{Binary Contests.} Our work is also related to binary contests \cite{GhoshK16, HaggiagOS22, LevyAS24, LevySA18, LevySA24, LevySR17, SarneL17}. In the binary contests model, each agent is characterized by a quality level \( q_i \) and a participation cost \( c_i \), with types \( (q_i, c_i) \) drawn from a known joint distribution \( F(\bq, \bc) \). An agent who enters the contest produces an outcome of quality \( q_i \) at a cost \( c_i \). A principal is tasked with allocating a total prize \( V \) among the participants based on their relative rankings. Each agent decides whether to participate by weighing the expected prize---determined by how their quality \( q_i \) ranks among other entrants---against their cost \( c_i \).  The key distinction in our model lies in the interpretation of \( q_i \)'s. Rather than representing a fixed quality level, \( q_i \) corresponds to a probability of success, introducing a different layer of stochasticity. As a result, an agent's expected payoff is based on the uncertainty of his success, together with that of other agents. In contrast, the binary contests framework involves randomness in sampling \( q_i \) and \( c_i \), but once these values are realized, the problem reduces to distributing a predetermined prize \( V \) among the highest-performing participants to maximize value extraction. Our setting diverges fundamentally in that there is no fixed prize budget; instead, costs are incurred based on the number of agents who ultimately succeed, making the reward structure inherently more contingent on realized outcomes.

\paragraph{Effort-Based Prize Splitting.}
Among contest-design models, the closest in spirit to \emph{anonymous/fair} reward rules is the effort-based prize-splitting literature, where the designer commits to a rule that divides a \emph{fixed prize budget} across participants as a function of their (typically rank-ordered) effort or performance, treating agents symmetrically. A recurring message is that concentrating rewards is often optimal: e.g., \citet{ArchakS09} and \citet{MoldovanuS06} identify conditions (risk-neutrality; linear/convex costs when the principal's value is the sum of effort) under which winner-takes-all is optimal, while \citet{ChawlaHS12} show winner-takes-all is optimal when prizes are committed ex ante and the principal only values the best submission; \citet{DiPalantinoV09} and related work highlight how reward levels shape participation. This perspective is complementary to ours. Prize splitting allocates a predetermined budget based on realized effort/ranks, whereas in our model outcomes are \emph{binary and stochastic}: an agent's ``performance'' is not a deterministic quality/effort level but a success event occurring with probability $q_i$. Accordingly, our rewards are naturally \emph{outcome-contingent} (e.g., paying agent $i$ only upon success) rather than budget-splitting rules over effort/ranks, and the added expressiveness of general contracts fundamentally changes the problem from choosing an optimal \emph{prize split} to choosing an \emph{outcome-contingent payment rule equal across successful agents}.

\section{Model and Preliminaries}\label{sec:model}

We study a multi-agent contract-design problem with observable individual outcomes. This puts us in line with \cite{CastiglioniM023}, and separates us from \cite{DuettingEFK23}.  
In our model a single principal interacts with $n$ agents. 
Each agent $i\in [n]$ is characterized by a tuple $(q_i, c_i) \in [0,1] \times \mathbb{R}_{\geq 0}$, where $q_i$ denotes the probability of success when exerting effort, and $c_i$ is the cost of exerting effort. When not exerting effort, the agent's probability of success and cost are both zero.
The principal can observe individual outcomes, but not whether an agent exerted effort or not.  
In particular, for each agent, there are two possible outcomes $\Omega = \{0,1\}$, where we interpret $0$ as failure and $1$ as success. We use $\omega_i \in \Omega$ to denote agent $i$'s outcome, and $\boldsymbol{\omega} = (\omega_1, \ldots, \omega_n) \in \Omega^n$ for the outcomes of all agents.
For a set $T \subseteq S$, let $Q_T(S)$ be the probability that out of all agents in $S$, exactly the set $T$ of agents succeeds:
\[
    \EqStyle Q_T(S) := \prod_{i\in T} q_i\prod_{i\in S\setminus T}(1-q_i).
\]

The principal cares about the number of agents that succeeded, 
and derives a certain reward $r \geq 0$ from each of them. Without loss of generality, we normalize the reward so that $r = 1$. Thus, when the set of agents $S \subseteq [n]$ exerts effort, the principal's expected reward is 
\[
   \EqStyle f(S) = \sum_{T\subseteq S} Q_T(S) \cdot |T| \cdot r = \sum_{T\subseteq S} Q_T(S) \cdot |T|.
\]

The social welfare achieved by a set of agents $S$ is given by $\textsc{SW}(S) = f(S) - \sum_{i \in S} c_i$. The maximum social welfare of an instance $\mathcal{I}$ is $\textsc{SW}(\mathcal{I}) = \max_{S} (f(S) - \sum_{i \in S} c_i)$.

\subsection{Contracts and Utilities}

To incentivize the agents to exert effort the principal employs a contract. In its most general form a contract is a tuple $\mathbf{t} = (t_1, \ldots, t_n)$ of payment functions $t_i: \{0,1\}^n \to \mathbb{R}$ for each agent $i$. The input to $t_i$ is a vector $\boldsymbol{\omega} \in \{0,1\}^n$ which indicates for each agent $i'$ whether agent $i'$ succeeded ($\omega_{i'} = 1$) or not ($\omega_{i'}= 0$).  
We interpret $t(T)$ as $t(\boldsymbol{\omega}(T))$ where $\boldsymbol{\omega}_i(T) = 1$ for $i \in T$ and $\boldsymbol{\omega}_i(T) = 0$ otherwise.

We may or may not impose limited liability. A contract satisfies \textit{limited liability} (LL) if payments are non-negative, that is $t_i(\omega) \ge 0$ for all $i$ and $\omega$. This requirement is often economically justified, and also serves to rule out trivial solutions to the contracting problem. An important thing to note is that with additional constraints on the contract (such as anonymity), the problem is non-trivial and therefore interesting even if we do not impose limited liability.

Each agent $i$ strategically decides to exert effort of not. Denoting by $S$ the set of agents that exert effort, agent $i$'s utility is 
\[
\EqStyle u_i(S,\mathbf{t}) =\sum_{T \subseteq S} Q_T(S) t_i(T) - \mathbf{1}\{i \in S\} \cdot c_i.
\]

The principal's utility in turn is 
\[
\EqStyle u_p(S,\mathbf{t}) = f(S) - \sum_{T \subseteq S} Q_T(S) \sum_{i \in [n]} t_i(T). 
\]

Each contract $\mathbf{t}$ induces a game among the agents, and we are interested in the (pure) Nash equilibria of that game. Formally, we say that a set $S$ is a (pure) Nash equilibrium of contract $\mathbf{t}$ (or that contract $\mathbf{t}$ incentives the set of agents $S$ to exert effort) if:
\begin{align*}
&u_i(S,\mathbf{t}) \geq u_i(S \setminus\{i\},\mathbf{t}) &&\text{for all $i \in S$, and} 
&u_i(S,\mathbf{t}) \geq u_i(S \cup \{i\},\mathbf{t}) &&\text{for all $i \not\in S$.} 
\end{align*}

The principal's (and our) goal is to find a contract $\mathbf{t}$ and a pure Nash equilibrium $S$ of $\mathbf{t}$ that maximizes the principal's utility  $u_p(S,\mathbf{t})$. 

\subsection{General vs.~Anonymous Contracts}

We classify contracts based on the amount of information the principal uses to determine payments and the constraints imposed on those payments. We define the hierarchy ranging from the most general to the most restrictive contracts. For each class we may or may not impose limited liability, as we shall see adding this requirement is without loss for general contracts, but for the more restricted classes there will be a distinction between the version with limited liability and without.

\paragraph{General Contracts.}
In the most general setting, the principal specifies a payment function $t_i: \{0,1\}^n \to \mathbb{R}$ for each agent $i$. This payment can depend on the full vector of outcomes $\omega$.  As observed in \cite{CastiglioniM023}, for general contracts, it is without loss of generality to make agent $i$'s payment contingent only on their own success, that is $t_i(\omega) > 0$ only if $\omega_i = 1$, and $0$ otherwise. This class allows the principal to extract full social welfare, even if we additionally impose limited liability. We have the following proposition (whose proof, provided for completeness, is deferred to \Cref{app:model}): 

\begin{restatable}[General contracts, \cite{CastiglioniM023}]{proposition}{propunrestricted}\label{prop:unrestricted}
For general contracts $\mathbf{t}$ with $t_i: \{0,1\}^n \to \mathbb{R}$ for each agent $i$ we have:
\begin{enumerate}
\item Every set of agents $S \subseteq [n]$ such that no agent $i \in S$ simultaneously has $q_i = 0$ and $c_i > 0$ can be incentivized by a general contract. 
\item If the set of agents $S \subseteq [n]$ can be incentivized by a general contract, then the optimal such contract sets $t_i(T) = \mathbf{1} \{i \in T\} \cdot c_i/q_i$ for $i \in S$ such that $q_i > 0$, and $t_i(T) = 0$ otherwise.
\item The optimal general contract for incentivizing set $S \subseteq [n]$ (if it can be incentivized) yields a principal's utility of $f(S) - \sum_{i \in S} c_i$.
\end{enumerate}
\end{restatable}

A perhaps counterintuitive property of optimal contracts is that they pay less capable agents more. This is true both ``ex post'' 
because
$c_i/q_i$ increases 
when
$c_i$ increases 
or
$q_i$ decreases, but also ``ex ante'' as an agent's expected payment is equal to $q_i \cdot c_i/q_i = c_i$ which increases as $c_i$ increases. 

In this work we focus on contracts that avoid the 
ex-post discrepancy
by paying all successful agents the same. We refer to such contracts as anonymous contracts.

\paragraph{Anonymous Contracts.}
An anonymous contract cannot depend on the identity, but only on the number of successful agents. Formally, let $T \subseteq S$ be the set of agents that succeed. The payment to agent $i$ is defined as $t_i(T) = \mathbf{1}\{i \in T\} \cdot w_{|T|}$, where $\mathbf{w} = (w_1, \dots, w_n)$ is a vector of payments fixed in advance. Here, $w_k$ represents the payment given to a successful agent conditional on exactly $k$ agents succeeding.

\paragraph{Uniform Anonymous Contracts.} This is a further restriction of anonymous contracts where the payment is constant regardless of the number of successes. Formally, $t_i(T) = \mathbf{1}\{i \in T\} \cdot w$, where $w \in \mathbb{R}$ is a single scalar value. In the notation of anonymous contracts, this corresponds to setting $w_j = w$ for all $j \in [n]$.

\subsection{Analyzing Anonymous Contracts} 

For the purpose of analyzing anonymous contracts, it will be useful to introduce some additional notation. First, since payments only depend on the number of successful agents, it will be useful to define the probability that exactly $t$ agents 
succeed when a given set $S$ of agents exerts effort:
\[
   \EqStyle Q_t(S) = \sum_{{T \subseteq S, |T|=t}} \left( \prod_{i \in T} q_i \cdot \prod_{i \in S \setminus T} (1-q_i) \right). 
\] 
Note that $Q_t(S) = 0$ for $t > |S|$.
Fixing payments $\mathbf{w}$, we can then express agent $i$'s utility as follows 
\begin{align}\label{eq:agent-utility}
    u_i(S, \bw) = \begin{cases}
    q_i \cdot \sum_{j \in [n]} Q_{j-1}(S_{-i}) \cdot w_j - c_i, &\text{for $i \in S$, and} \\
    0 &\text{otherwise}.
    \end{cases}
\end{align}
Similarly, the principal's utility becomes
\begin{align}\label{eq:principal-utility}
    \EqStyle u_p(S, \bw) = f(S) - \sum_{i \in S} \left( q_i \cdot \sum_{j \in [n]} Q_{j-1}(S_{-i}) \cdot w_j \right).
\end{align}
In order to maximize her expected utility, the principal thus aims to solve the following program: 
\begin{align}
    &\max_{S \subseteq [n], \bw \geq \bzero} ~ f(S) - \EqStyle \sum_{i \in S} \left( q_i \cdot \sum_{j \in [n]} Q_{j-1}(S_{-i}) \cdot w_j \right) \label{eq:principal-program1}\\
    &\quad\text{ s.t. }\qquad\; \EqStyle q_i \cdot \sum_{j \in [n]} Q_{j-1}(S_{-i}) \cdot w_j \geq c_i \qquad \forall i \in S \label{eq:principal-program2}\\
    &\quad\qquad\qquad \EqStyle q_i \cdot \sum_{j \in [n]} Q_{j-1}(S) \cdot w_j \leq c_i \qquad\forall i \notin S \label{eq:principal-program3}.
\end{align}

Here \eqref{eq:principal-program1} is simply the principal's utility $u_p(S,\mathbf{w})$, while \eqref{eq:principal-program2} and \eqref{eq:principal-program3} capture that $S$ should be a pure Nash equilibrium of $\mathbf{w}$. We succinctly denote as $u_p(S(\bw))$ the highest utility the principal can get by offering contract $\bw$, where $S(\bw)$ is the maximizing set for the specific $\bw$. As it turns out, analyzing this mathematical program is much more challenging then solving the unconstrained contract problem.

\section{Structural Insights on Anonymous Contracts}\label{sec:insights}

The primary objective of this section is to establish crucial insights into the structure of anonymous contracts. Namely, we first establish that every anonymous contract admits a (pure) Nash equilibria. Then, we show that the principal's utility gap between worst and best equilibria can be unbounded.

\subsection{Existence of (Pure) Nash Equilibria}\label{sec:pne}

The goal of this section is to prove the following theorem:
\begin{restatable}{theorem}{thmequilibrium}\label{thm:equilibrium}
    For all payment vectors $\bw \in \R^n$, the anonymous contract defined by $\bw$ admits a (pure) Nash equilibrium. Moreover, this holds for any principal's reward function.
\end{restatable}

Observe that the above theorem holds for all payment vectors $\bw \in \R^n$ (which can have positive or negative components), i.e., whether or not the principal has to satisfy limited liability.

To establish \Cref{thm:equilibrium}, we first prove that the anonymous contract defined by $\bw$ has the finite improvement property. 
By \citep[Lemma 2.3]{MondererShapley96}, this implies that the contract has a (pure) Nash equilibrium, for every $\bw$.

\begin{restatable}{lemma}{lempnecyle}\label{lem:pne-cycle}
    For all payment vectors $\bw \in \R^n$, if the finite improvement property holds, then a (pure) Nash equilibrium implementing $\bw$ exists.
\end{restatable}

The proof of this fact is standard, but we provide one for completeness in Appendix~\ref{app:equilibria}.\footnote{Note that by \citep[Lemma 2.5]{MondererShapley96} having the finite improvement property is equivalent to admitting a generalized ordinal potential.} To conclude the proof of Theorem~\ref{thm:equilibrium}, we then establish the finite improvement property.

\begin{restatable}{lemma}{lemnocyle}\label{lem:nocycle}
    For all payment vectors $\bw \in \R^n$, the anonymous contract defined by $\bw$ has the finite improvement property.
\end{restatable}

\paragraph{Notation.} Let us fix a payment vector $\bw$, so that all the subsequent notation implicitly depends on this fixed $\bw$. Moreover, let $S(0)$ be an arbitrary set of agents that exert effort. At step $t \geq 1$ an agent $i(t)$ that is not best-responding in $S(t-1)$ updates their action leading to $S(t)$ which is 
\begin{enumerate}
\item[(i)] $S(t) = S(t-1) \cup \{i(t)\}$ if agent $i(t)$ starts to exert effort at step $t$; 
\item[(ii)] $S(t) = S(t-1) \setminus \{i(t)\}$ if agent $i(t)$ stops exerting effort at step $t$.
\end{enumerate}

We call any sequence of such sets $S(t)$ a best-response sequence, and say that is finite if at some finite step $T$ all agents best respond. The \emph{finite improvement property} asserts that all best-response sequences are finite. In contrast, every infinite best-response sequence has to contain a cycle.

In addition, for any $S$, let us denote by $g_i(S)$ the expected payment agent $i$ gains by joining a set $S \setminus \{i\}$ of agents forming $S$, and by $\ell_i(S)$ the expected payment agent $i$ loses
by leaving a set $S$ of agents forming $S \setminus \{i\}$. Formally, we have:
\begin{align*}
&g_i(S) = u_i(S) + c_i &&\text{if $i$ joins $S \setminus \{i\}$}\\
&\ell_i(S) = u_i(S) + c_i &&\text{if $i$ leaves $S$}.
\end{align*}
Specifically, in a best-response sequence, considering two subsequent sets, $S(t-1)$ and $S(t)$,
\begin{align*}
&g_{i(t)}(S(t)) > c_{i(t)} &&\text{if $i(t)$ joins $S(t-1)$ at step $t$ forming $S(t)$, and }\\
&\ell_{i(t)}(S(t-1)) < c_{i(t)} &&\text{if $i(t)$ leaves $S(t-1)$ at step $t$ forming $S(t)$}.
\end{align*}
We refer to the former as an incrementing deviation, and the latter as a decrementing deviation.

\paragraph{Proof strategy.} The rest of the section is dedicated to proving Lemma~\ref{lem:nocycle}, which implies Theorem~\ref{thm:equilibrium}, together with Lemma~\ref{lem:pne-cycle}. Specifically, we need to show that there cannot be any cycle (and, hence, the finite improvement property is satisfied). We proceed by contradiction: Suppose there exists a best-response sequence forming a cycle. We will prove that the total cost over increasing deviations ($\ID$) equals the total cost over decreasing deviations ($\DD$): $C := \sum_{t \in \ID} c_{i(t)} = \sum_{t \in \DD} c_{i(t)}$.
This leads to the inequality:
\[
\EqStyle \sum_{t \in \ID} g_{i(t)}(S(t)) > C > \sum_{t \in \DD} \ell_{i(t)}(S(t)) \Longrightarrow
\sum_{t \in \ID} g_{i(t)}(S(t)) - \sum_{t \in \DD} \ell_{i(t)}(S(t)) > 0.
\]
At the same time, we will show that this latter difference is actually $0$, leading to a contradiction. Formally, we have the following claims:

\begin{restatable}{claim}{clcostcycle}\label{cl:cost-cycle}
    In a best-response sequence forming a cycle, it holds that
    \[
       \EqStyle \sum_{t \in \ID} c_{i(t)} = \sum_{t \in \DD} c_{i(t)}.
    \]
\end{restatable}

We also have the following useful property which asserts how the gain arising from an incrementing deviation and the loss of an immediately successive decrementing deviation change.

\begin{restatable}{claim}{clrewardreduction}\label{cl:reward-reduction}
    For all $\bw \in \R^n$ and $S \subseteq [n]$, if $i^+$ causes an incrementing deviation from $S \setminus \{i^+\}$ to $S$ and $i^-$ causes a decrementing deviation from $S$ to $S \setminus \{i^-\}$, then
    \[
        g_{i^+}(S) - \ell_{i^-}(S) = g_{i^+}(S \setminus \{i^-\}) - \ell_{i^-}(S \setminus \{i^+\}).
    \]
\end{restatable}

We now partition the cycle into \emph{peak walks} $P_1, \ldots, P_m$, that is walks that contain incrementing deviations until a decrementing deviation, and then decrementing deviations until (and excluding) the next incrementing deviation. For a peak walk $P_r$, $\ID(P_r), \DD(P_r)$ are respectively the sets of incrementing and decrementing deviations within that peak walk.

\begin{restatable}{claim}{clpeakwalk}\label{cl:peak-walk}
    Within a peak walk $P_r$, it holds that
    \begin{align*}
        &\EqStyle \sum_{t \in \ID(P_r)} g_{i(t)}(S(t)) - \sum_{t \in \DD(P_r)} \ell_{i(t)}(S(t)) \\
         = & \EqStyle \sum_{t \in \ID^\star(P_r)} g_{i(t)}\left(S(t) \setminus \bigcup_{t^\prime \in \DD(P_r)} \{i(t^\prime)\}\right) - \sum_{t \in \DD^\star(P_r)} \ell_{i(t)}\left(S(t) \setminus \bigcup_{t^\prime \in \ID(P_r)} \{i(t^\prime)\}\right),
     \end{align*}
     where $\ID^\star(P_r), \DD^\star(P_r)$ are, respectively, the sets of all incrementing and decrementing deviations within peak walk $P_r$, excluding those deviations such that agents $i(t)$ cause both an incrementing and a decrementing deviation within $P_r$.
\end{restatable}

We defer the proof of Claim~\ref{cl:peak-walk} to Appendix~\ref{app:equilibria}. Here, we exemplify its meaning through Figure~\ref{fig:cycle}. In particular, consider the first peak walk, going from set $\{2,3\}$ until set $\{2,3,4,5\}$. In this peak walk, the first leaving agent, agent 1, has initially an expected payment he is giving up equal to $\ell_1(\{1,2,3,4,5\})$. All other agents in this peak walk have caused incrementing deviations, and so, by Claim~\ref{cl:reward-reduction}, the expected payment of agent 1 reduces with that of agent 5, namely $\ell_1(\{1,2,3,4,5\})$ becomes $\ell_1(\{1,2,3,4\})$ and, symmetrically, $g_5(\{1,2,3,4,5\})$ becomes $g_5(\{2,3,4,5\})$. Then, the same happens with agent 4, where  $\ell_1(\{1,2,3,4\})$ becomes $\ell_1(\{1,2,3\})$ and, symmetrically, $g_4(\{1,2,3,4\})$ becomes $g_4(\{2,3,4\})$. This continues until the first agent in the peak walk causing an incrementing deviation, which in this case is precisely agent $1$. Indeed, in this case, the expected payment $\ell_1(\{1,2,3\})$ given up by agent 1 leaving at the end of the peak walk and remaining from previous reductions, cancels out with that of agent 1 that has just joined the peak walk, i.e., $g_1(\{1,2,3\})$. For all those agents that have not joined or left within the same peak walk, we say that their expected payments are in \emph{irreducible form within the peak walk}.

\begin{figure}[htbp]
\centering
\resizebox{0.8\linewidth}{!}{
\tikzset{every picture/.style={line width=0.75pt}}        

\begin{tikzpicture}[x=0.75pt,y=0.75pt,yscale=-1,xscale=1, draw=blue]

\draw[draw=blue]    (51,170) .. controls (2.98,162.16) and (12.58,149.52) .. (51.58,146.19) ;
\draw [shift={(54,146)}, rotate = 175.82, line width=0.08, draw=blue, fill=blue]
      (8.93,-4.29) -- (0,0) -- (8.93,4.29) -- cycle ;

\draw[draw=blue]    (79,138) .. controls (30.98,130.16) and (50.19,124.24) .. (89.56,121.18) ;
\draw [shift={(92,121)}, rotate = 175.82, line width=0.08, draw=blue, fill=blue]
      (8.93,-4.29) -- (0,0) -- (8.93,4.29) -- cycle ;

\draw[draw=blue]    (128,112) .. controls (79.98,104.16) and (103.03,99.2) .. (142.56,96.18) ;
\draw [shift={(145,96)}, rotate = 175.82, line width=0.08, draw=blue, fill=blue]
      (8.93,-4.29) -- (0,0) -- (8.93,4.29) -- cycle ;

\draw[draw=blue]    (189,87) .. controls (178.17,73.21) and (236.22,49.72) .. (239.87,108.27) ;
\draw [shift={(240,111)}, rotate = 268.15, line width=0.08, draw=blue, fill=blue]
      (8.93,-4.29) -- (0,0) -- (8.93,4.29) -- cycle ;

\draw[draw=red]    (259,113) .. controls (256.06,109.08) and (236.79,93.64) .. (276.49,97.72) ;
\draw [shift={(279,98)}, rotate = 186.63, line width=0.08, draw=red, fill=red]
      (8.93,-4.29) -- (0,0) -- (8.93,4.29) -- cycle ;

\draw[draw=red]    (325,87) .. controls (293.64,83.08) and (311.26,72.44) .. (350.57,69.19) ;
\draw [shift={(353,69)}, rotate = 175.82, line width=0.08, draw=red, fill=red]
      (8.93,-4.29) -- (0,0) -- (8.93,4.29) -- cycle ;

\draw[draw=red]    (421,58) .. controls (410.17,44.21) and (468.22,20.72) .. (471.87,79.27) ;
\draw [shift={(472,82)}, rotate = 268.15, line width=0.08, draw=red, fill=red]
      (8.93,-4.29) -- (0,0) -- (8.93,4.29) -- cycle ;

\draw[draw=red]    (493,86) .. controls (482.17,72.21) and (540.22,48.72) .. (543.87,107.27) ;
\draw [shift={(544,110)}, rotate = 268.15, line width=0.08, draw=red, fill=red]
      (8.93,-4.29) -- (0,0) -- (8.93,4.29) -- cycle ;

\draw[draw=red]    (554,112) .. controls (543.17,98.21) and (588.6,78.6) .. (591.88,137.27) ;
\draw [shift={(592,140)}, rotate = 268.15, line width=0.08, draw=red, fill=red]
      (8.93,-4.29) -- (0,0) -- (8.93,4.29) -- cycle ;

\draw[draw=red]    (585,163) .. controls (467.59,206.78) and (161.09,216.9) .. (57.55,192.37) ;
\draw [shift={(56,192)}, rotate = 13.77, line width=0.08, draw=red, fill=red]
      (8.93,-4.29) -- (0,0) -- (8.93,4.29) -- cycle ;

\draw (32,172.4) node [anchor=north west][inner sep=0.75pt] {$\{2,3\}$};
\draw (58,139.4) node [anchor=north west][inner sep=0.75pt] {$\{1,2,3\}$};
\draw (92,114.4) node [anchor=north west][inner sep=0.75pt] {$\{1,2,3,4\}$};
\draw (145,89.4) node [anchor=north west][inner sep=0.75pt] {$\{1,2,3,4,5\}$};
\draw (215,115.4) node [anchor=north west][inner sep=0.75pt] {$\{2,3,4,5\}$};
\draw (277,88.4) node [anchor=north west][inner sep=0.75pt] {$\{2,3,4,5,6\}$};
\draw (355,62.4) node [anchor=north west][inner sep=0.75pt] {$\{2,3,4,5,6,7\}$};
\draw (440,88.4) node [anchor=north west][inner sep=0.75pt] {$\{2,3,5,6,7\}$};
\draw (502,115.4) node [anchor=north west][inner sep=0.75pt] {$\{2,3,5,7\}$};
\draw (555,140.4) node [anchor=north west][inner sep=0.75pt] {$\{2,3,7\}$};
\draw (21,131.4) node [anchor=north west][inner sep=0.75pt] {$+$};
\draw (64,101.4) node [anchor=north west][inner sep=0.75pt] {$+$};
\draw (120,78.4) node [anchor=north west][inner sep=0.75pt] {$+$};
\draw (257,78.4) node [anchor=north west][inner sep=0.75pt] {$+$};
\draw (329,51.4) node [anchor=north west][inner sep=0.75pt] {$+$};
\draw (210,51.4) node [anchor=north west][inner sep=0.75pt] {$-$};
\draw (438,22.4) node [anchor=north west][inner sep=0.75pt] {$-$};
\draw (511,52.4) node [anchor=north west][inner sep=0.75pt] {$-$};
\draw (562,79.4) node [anchor=north west][inner sep=0.75pt] {$-$};
\draw (306,180.4) node [anchor=north west][inner sep=0.75pt] {$-$};

\end{tikzpicture}
}
\caption{We portray a cycle composed of two peak walks: One going from set $\{2,3\}$ until set $\{2,3,4,5\}$ (in blue), and the other from set $\{2,3,4,5\}$ until set $\{2,3\}$ (in red).}
\label{fig:cycle}
\end{figure}
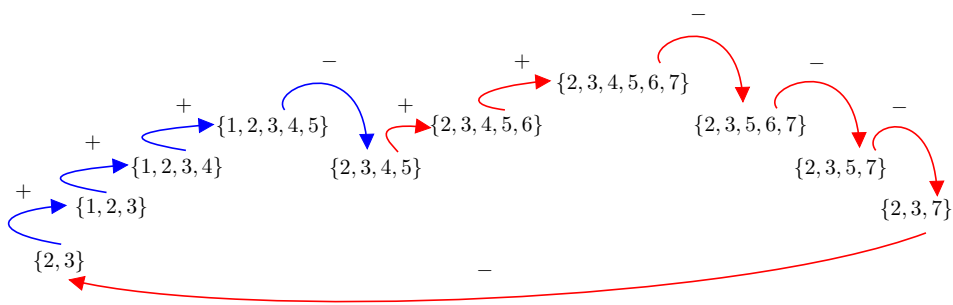

We would like to stress that Claim~\ref{cl:peak-walk} only shows how, if an agent joins and leaves within the same peak walk, then his expected payments will reduce until canceling due to the aforementioned procedure. In general, however, an agent need not join and leave within the same peak walk. Lemma~\ref{lem:nocycle} handles such cases and shows that, iteratively, this reduction-until-cancellation behavior still occurs. Then, if a best-response sequence were to contain a cycle, this would yield a contradiction. We proceed with the proof sketch of Lemma~\ref{lem:nocycle}, to show that if a best-response sequence contained a cycle, this would yield a contradiction. The full proof is deferred to Appendix~\ref{app:equilibria}.

\begin{proof}[Proof Sketch of Lemma~\ref{lem:nocycle}]
    To prove that there cannot be any cycle (and, hence, the finite improvement property is satisfied), we proceed by contradiction. Suppose there exists a best-response sequence forming a cycle. By Claim~\ref{cl:cost-cycle}, the total cost over increasing deviations equals the total cost over decreasing deviations: $C := \sum_{t \in \ID} c_{i(t)} = \sum_{t \in \DD} c_{i(t)}$.
This leads to the inequality:
\[
\EqStyle \sum_{t \in \ID} g_{i(t)}(S(t)) > C > \sum_{t \in \DD} \ell_{i(t)}(S(t)) \Longrightarrow
\sum_{t \in \ID} g_{i(t)}(S(t)) - \sum_{t \in \DD} \ell_{i(t)}(S(t)) > 0.
\]
However, we aim to show that this latter difference is actually $0$, leading to a contradiction. The deviation sums can be expressed as:
\begin{align}
  \EqStyle \sum_{r \in [m]} \left(\sum_{t \in \ID(P_r)} g_{i(t)}(S(t)) - \sum_{t \in \DD(P_r)} \ell_{i(t)}(S(t))\right) > 0.\label{eq:cycle-contradiction}
\end{align}
By Claim~\ref{cl:peak-walk}, we reduce the problem to analyzing only irreducible deviations within a peak, meaning agents that join in one peak walk and leave in another.

We consider an arbitrary agent \( i^\star \) who (i) joins at time \( t^+_a \) in peak walk \( P_a \) (incrementing deviation), (ii) leaves at time \( t^-_b \) in peak walk \( P_b \) (decrementing deviation), where \( P_a \neq P_b \), (iii) has no other deviation in between.

We analyze the reward differences across peak walks and show that they cancel. Specifically, we show that
\[
    g_{i^\star}\left(S(t^+_a) \setminus X \right) - \ell_{i^\star}\left(S(t^-_b) \setminus Y \right) = 0,
\]
where \( X \) is the set of agents that left before \( i^\star \) and \( Y \) is the set of agents that joined after \( i^\star \). 

To do so, we first consider two consecutive peak walks \( P_r \) and \( P_{r+1} \). Let \( t^+_1(P_r) \) be the time of the last agent joining in \( P_r \), and let \( t^-_1(P_{r+1}) \) be the first agent leaving in \( P_{r+1} \). We show that the set of agents present just before \( t^+_1(P_r) \) is identical to the set just before \( t^-_1(P_{r+1}) \), leading to:
\[
    g_{i(t^+_1(P_r))}(S(t^+_1(P_r)) \setminus \text{previous leavers}) = \ell_{i(t^-_1(P_{r+1}))}(S(t^-_1(P_{r+1})) \setminus \text{previous joiners}).
\]
This process repeats iteratively across all peak walks, ensuring that each joining agent's reward cancels with their corresponding leaving reward. Since this cancellation holds for all agents,
 \[
 \EqStyle \sum_{r \in [m]} \left(\sum_{t \in \ID(P_r)} g_{i(t)}(S(t)) - \sum_{t \in \DD(P_r)} \ell_{i(t)}(S(t))\right) = 0.
 \] 
which contradicts our earlier inequality in~\eqref{eq:cycle-contradiction}.
Thus, no best-response sequence can form a cycle, completing the proof sketch. 
\end{proof}

\subsection{Gap between PNEs under Anonymous Contracts}\label{app:pne-observations}

We conclude this section by making a few observations about the structure and complexities of (pure) Nash equilibria for anonymous contracts with or without limited liability. First, we next show that, under limited liability, the converse of Theorem~\ref{thm:equilibrium} does not hold:
\begin{proposition}
    Under limited liability, there exists $S$  such that all $\bw \in \R^n_+$ are infeasible.
\end{proposition}
\begin{proof}
    The statement follows because we could consider an instance with two agents $(q_1,c_1) = (\varepsilon, \frac{\varepsilon^2}{2})$ and $(q_2,c_2) = (\frac{1}{2}, \frac{1}{3})$, where we want to incentivize $S = \{2\}$. Then, the following must hold:
    \begin{align*}
        q_2w_1 \geq c_2 ~\text{and}~ (1-q_2)q_1w_1 + q_1q_2w_2 \leq c_1 \Longleftrightarrow w_1 \geq {2}/{3} ~\text{and}~ w_1 + w_2 \leq \varepsilon.
    \end{align*}
    As long as $\bw \in \R^n_+$, these are impossible to satisfy for any $\varepsilon \leq \tfrac{2}{3}$. 
\end{proof}

By Theorem~\ref{thm:equilibrium}, every $\bw \in \R^n$ induces a set $S(\bw)$ of agents which constitutes a (pure) Nash equilibrium. We next observe that there are $\bw \in \R^n_+$ such that the gap in principal utility between the best and the worst equilibrium is unbounded. 

\begin{proposition}\label{prop:unbounded}
    For any $\varepsilon > 0$, there exists an instance $\mathcal{I}$ and a contract $\bw \in \R^n_+$ such that the anonymous contract game defined by $\bw$ admits at least two (pure) Nash equilibria $S$ and $S'$ such that
    \[
        {u_p(S, \bw)}/{ u_p(S', \bw)} = ({1-\varepsilon})/{2\varepsilon}
    \]
    and thus $u_p(S, \bw)/u_p(S', \bw) \rightarrow \infty$ as $\varepsilon \rightarrow 0.$
\end{proposition}

\begin{proof}    
    Consider an instance with two agents with $(q_1,c_1)=(2\varepsilon, \varepsilon)$, $(q_2,c_2)=(1-\varepsilon, \frac{1-\varepsilon}{2})$, and assume $\bw=(w_1,0)$. We show that this payment structure is implementable in two equilibria, but the two equilibria yield vastly different utilities for the principal. Note also that, for all singleton sets, the payment structure $\bw=(w_1,0)$ is optimal, as it is never profitable to have $w_2 > 0$.
    
    First, suppose we want to incentivize set $S=\{1\}$. Then, the following constraints must hold:
    \begin{align*}
        q_1w_1 \geq c_1 ~\text{and}~ (1-q_1)q_2w_1 \leq c_2 \Longleftrightarrow {1}/{2} \leq w_1 \leq {1}/\left(({2(1-2\varepsilon)}\right).
    \end{align*}
    Conversely, suppose we want to incentivize $S=\{2\}$. Then, we must have
    \begin{align*}
        q_2w_1 \geq c_2 ~\text{and}~ (1-q_2)q_1w_1 \leq c_1 \Longleftrightarrow {1}/{2} \leq w_1 \leq {1}/({2\varepsilon}).
    \end{align*}
    In both cases, setting $w_1 = \tfrac{1}{2}$ is the minimal payment that satisfies the respective constraints.  
    Since $u_p\left(\{1\}, \left(\tfrac{1}{2}, 0\right)\right) =\varepsilon$ and $u_p\left(\{2\}, \left(\tfrac{1}{2}, 0\right)\right) = \tfrac{1-\varepsilon}{2}$ the claim follows.  
\end{proof}

\subsection{Uniqueness of PNE under Uniform Anonymous Contracts}
\label{sec:uniqueness_of_pne_under_uniform_anonymous_contracts}

We conclude this section by identifying a desirable property of uniform anonymous contracts, namely that they admit unique equilibria. This is a desirable property, because it immediately implies that any approximation guarantee for uniform anonymous contracts is a robust approximation guarantee that applies under \emph{any} equilibrium.

\begin{proposition}\label{prop:equal}
     Consider an instance $\cI = (\bq,\bc)$ and assume indices are ordered such that $c_i/q_i$ is non-decreasing. Then any %
     $w = w_1 = w_2 = \ldots = w_n$ will incentivize a unique set $S$. Moreover, $S = [k]$ for some $k\in \{0\} \cup[n]$, and
     \[
       \EqStyle \EQ(\cI) = \max_{k \in [n]} \left(1 - {c_k}/{q_k}\right) \cdot \sum_{i \in [k]} q_i. 
    \]
\end{proposition}

\begin{proof} 
    First, in a uniform anonymous contract, constraints \eqref{eq:principal-program2} and \eqref{eq:principal-program3}, respectively become
    \begin{align*}
        &\EqStyle \sum_{j \in [k]} Q_{j-1}(S \setminus\{i\}) \cdot {w_j}\geq {c_i}/{q_i} \Longrightarrow w \cdot \sum_{j \in [k]} Q_{j-1}(S \setminus\{i\}) \geq {c_i}/{q_i} \qquad &\forall i \in S,\\
        &\EqStyle \sum_{j \in [k+1]} Q_{j-1}(S) \cdot {w_j} \leq{c_i}/{q_i} \Longrightarrow w \cdot \sum_{j \in [k+1]} Q_{j-1}(S) \leq {c_i}/{q_i} \qquad &\forall i \notin S.
    \end{align*}
    Since $w$ is fixed and $\sum_{j \in [k]} Q_{j-1}(S \setminus\{i\}) = 1$ and $\sum_{j \in [k+1]} Q_{j-1}(S) = 1$ by definition, we may simplify the program formulation in \eqref{eq:principal-program1}-\eqref{eq:principal-program3} as
    \begin{align*}
        \EqStyle \max \left\{(1-w) \cdot \sum_{i \in S} q_i \;\big\lvert\; \max_{i \in S} {c_i}/{q_i} \leq w \leq \min_{i \notin S} {c_i}/{q_i}, S \subseteq [n] \right\}.
    \end{align*}
    Hence, the incentivizable set $S$ is not arbitrary once its size is fixed to $k$. Indeed, recalling that agents are sorted by $\frac{c_i}{q_i}$, the constraints dictate $w \geq \frac{c_k}{q_k}$, the largest density in $S$. This means that all the agents with lower density than $\frac{c_k}{q_k}$ are unavoidably incentivized. Moreover, all agents with higher or equal density to $\frac{c_{k+1}}{q_{k+1}}$ are never incentivized.\footnote{In case of ties, we break them lexicographically.} This means that the above program is further simplified to
     \begin{align*}
        \EqStyle \max \left\{(1-w) \cdot \sum_{i \in [k]} q_i \;\big\lvert\; {c_k}/{q_k} \leq w \leq {c_{k+1}}/{q_{k+1}} \right\},
    \end{align*}
     which shows the first part. For what concerns the second part, we notice, from the above program, that $w$ always belongs to a non-empty interval. Hence, the optimal uniform anonymous contract %
     achieves a utility of
    \[
        \EqStyle \max_{k \in [n]} \left(1 - {c_k}/{q_k}\right) \cdot \sum_{i \in [k]} q_i. \qedhere
    \]
\end{proof}

\begin{remark}
    Note that in the above proposition, one need not require limited liability explicitly. Indeed, $w$ is already constrained to be nonnegative by the constraints $\max_{i \in S} \tfrac{c_i}{q_i} \leq w \leq \min_{i \notin S} \tfrac{c_i}{q_i}$.
\end{remark}

\section{Anonymous Contracts With Limited Liability}\label{sec:limited_liabbility}

In this section, we examine the performance of anonymous contracts with limited liability in general settings. Our results identify the maximum spread in success probabilities as the driving force between the ability of anonymous contracts to approximate optimal unrestricted contracts.

\subsection{Lower Bound Governed by the Success-Probability Spread}
\label{sec:general_instances}

We start with an impossibility result that applies to any anonymous contract. This impossibility shows a logarithmic lower bound in terms of the maximum spread of success probabilities, and shows that in the worst-case the loss from using an anonymous contract rather than a general one can be linear in the number of agents.

\begin{restatable}[Negative Result for General Instances, Anonymous Contracts with LL]{theorem}{negativegeneral}\label{thm:negativegeneral}
Let $\OPT(\mathcal{I})$ denote the principal's utility from the optimal anonymous contract with limited liability on instance $\mathcal{I}$.
For any $Q \geq 1$ and $n$, there exists an instance $\cI$ with $n$ agents with success probabilities $q_1, \ldots, q_n$ such that $(\max_i q_i)/(\min_i q_i) \leq Q$ and:
\[\SW(\cI)=\Omega(\min\{n,\log Q\}) \cdot \OPT(\cI).\]
\end{restatable}

In order to prove \Cref{thm:negativegeneral}, we define an instance with $n$ agents as follows. 
By scaling $Q$ appropriately, we can assume without loss of generality that $Q \leq 2^n$.
  Let $\ell = \lfloor \log_2Q\rfloor$. 
   For $i\in [\ell ]$, the success probabilities and costs are given by:
    \begin{align*}
        q_1 = \frac{1}{2^{2\ell }}, \quad q_2 = \frac{1}{2^{2\ell -1}}, \quad \ldots, \quad q_\ell = \frac{1}{2^{\ell +1}}, \quad \text{and}\quad c_i = q_i - \varepsilon \quad \text{ where } \quad \varepsilon = \frac{1}{2^{2\ell +1}}.
    \end{align*}
    If $\ell <n$, for all $i = \ell +1,\ldots,n$, let $q_i = 1/2^{\ell+2}$ and $c_i = 1/2^{\ell +1}$.
    Note that we may assume without loss of generality that $\OPT(\cI)$ does not incentivize any agent $i>\ell$, since for these agents the success probability is strictly smaller than the cost.

The following key lemma shows that if the  the principal incentivizes two agents $i < j$, then the expected  transfer for the smaller agent $i$ is quite large compared to the expected reward $q_i$ generated by that agent.

\begin{lemma}[Lower Bound on Expected Transfer]\label{lem:lower_bound_on_expected_transfer}
Fix an anonymous contract $\bw = (w_1, \ldots, w_{\ell})$ and a corresponding Nash equilibrium $S \subseteq [\ell]$. Let $j \in S$ denote the largest index in $S$. Then, for any $i \in S \setminus \{j\}$ it holds that 
\begin{align*}
    \EqStyle q_i \cdot \sum_{r \leq |S|} Q_{r-1}(S_{-i}) \cdot w_r \geq  \left( {q_i \cdot (1-q_j) \cdot  (q_j - \varepsilon)} \right) / \left( {q_j \cdot (1-q_i)} \right).
\end{align*}
\end{lemma}
\begin{proof}
    Fix any agent $i \in S \setminus \{j\}$.
    We will prove a lower bound on the expected transfer received by agent $i$ from the principal under contract $\boldsymbol{w}$. We rewrite the expected transfer received by agent $i$ in terms of two auxiliary variables. Define
    \begin{align*}
        \EqStyle x_1 = \sum_{r \in [|S|]} Q_{r-2}(S_{-i,j}) \cdot w_{r}
        \quad\text{and}\quad
         x_2 = \sum_{r \in [|S|]} Q_{r-1}(S_{-i,j}) \cdot w_{r} 
    \end{align*}
    Then the expected transfer to agent $i$ can be written as
    \begin{align*}
    \EqStyle q_i \cdot \sum_{r \in [|S|]} Q_{r-1}(S_{-i}) \cdot w_{r}        &= q_i \cdot q_j \cdot x_1 + q_i \cdot (1-q_j) \cdot x_2 \eqcolon g_i(x_1, x_2) .
    \end{align*}
    Similarly, the expected transfer to agent $j$ can be written using the same $(x_1,x_2)$ as
    \begin{align*}
         \EqStyle q_j \cdot \sum_{r \in [|S|]} Q_{r-1}(S_{-j}) \cdot w_{r}        &= q_i \cdot q_j \cdot x_1 + (1-q_i) \cdot q_j \cdot x_2 \eqcolon g_j(x_1, x_2).
    \end{align*}
 
By the incentive compatibility constraints, we have:
    \begin{align}
        &0 \leq u_i(S, \bw) = g_i(x_1, x_2)  - c_i =  g_i(x_1,x_2) - (q_i - \varepsilon) \quad\Longrightarrow\quad g_i(x_1,x_2) \geq q_i - \varepsilon \label{ineq:xei}\\
        &0 \leq u_j(S, \bw) = g_j(x_1, x_2) - c_j = 
        g_j(x_1,x_2) - ( q_j - \varepsilon) \quad\Longrightarrow\quad g_j(x_1,x_2) \geq q_j - \varepsilon \label{ineq:xej}.
    \end{align}

Define $x^\star = (x_1^\star, x_2^\star) = \left(1 - {\varepsilon}/({q_i \cdot q_j}), 1\right)$. 
Note that:
    \begin{align}
        g_i(x_1^\star, x_2^\star) = q_i \cdot q_j \cdot x^\star_1 + q_i \cdot (1-q_j) \cdot x^\star_2 &= q_i \cdot q_j - \varepsilon + q_i \cdot (1-q_j) = q_i - \varepsilon
        \label{ineq:xstarei}\\
        g_j(x_1^\star, x_2^\star) = q_i \cdot q_j \cdot x^\star_1 + q_j \cdot (1-q_i) \cdot x^\star_2 &= q_i \cdot q_j - \varepsilon + q_j \cdot (1-q_i) = q_j - \varepsilon.
        \label{ineq:xstarej}
    \end{align}
        
Moreover, observe that:
    \begin{align*}
        x_1^\star = 1-\varepsilon/(q_i \cdot q_j) \leq 1-\varepsilon/q_\ell^2 = 1- (1/2^{2\ell +1})/(1/2^{2\ell +2}) = 1-2 < 0.
    \end{align*}
    
Since $x_1 \geq 0$ due to the limited liability assumption, there exists $0 \leq \gamma \leq 1$ such that $\gamma \cdot x_1 + (1 - \gamma) \cdot x_1^\star = 0$. Define $\widetilde{x} = (\widetilde{x}_1, \widetilde{x}_2) = \gamma \cdot x + (1 - \gamma) \cdot x^\star$. Then by \Cref{ineq:xej,ineq:xstarej}:
     \begin{align*}
         g_j(\widetilde{x}_1, \widetilde{x}_2) = \gamma \cdot g_j(x_1, x_2) + (1-\gamma) \cdot g_j(x^\star_1, x^\star_2) \geq q_j - \varepsilon.
    \end{align*}
Since $\widetilde{x}_1 = 0$ by choice of $\gamma$, we have $g_j(\widetilde{x}_1, \widetilde{x}_2) = q_j \cdot (1-q_i) \cdot \widetilde{x}_2$, which implies:
    \begin{align*}
        \widetilde{x}_2 \geq \left( {q_j - \varepsilon} \right) / \left( {q_j \cdot (1-q_i)} \right). 
    \end{align*}
Now, using \Cref{ineq:xei,ineq:xstarei}, we get:
    \begin{align*}
    g_i(x_1, x_2) &\geq
         \gamma \cdot g_i(x_1,x_2) + (1-\gamma) \cdot g_i(x^\star_1,x^\star_2) 
         =  g_i(\widetilde{x}_1, \widetilde{x}_2) 
         = q_i \cdot (1-q_j) \cdot \widetilde{x}_2 
         \\
         &\geq \left( {q_i \cdot (1-q_j) \cdot  (q_j - \varepsilon)} \right) / \left( {q_j \cdot (1-q_i)} \right). \qedhere
    \end{align*}
\end{proof}

\Cref{thm:negativegeneral} follows from \Cref{lem:lower_bound_on_expected_transfer}:

\begin{proof}
   Fix an anonymous contract $\bw = (w_1, \ldots, w_\ell)$ and a corresponding Nash equilibrium $S$. We will prove that $u_p(S, \bw) \leq (4/\ell) \cdot \SW$. Let $j \in S$ denote the largest index in $S$.
Using \Cref{lem:lower_bound_on_expected_transfer}, we can now bound the principal's utility of contract $w$ as follows:
    \begin{align*}
        u_p(S, \bw) 
        &= \sum_{i \in S} q_i - \sum_{i \in S} q_i \cdot \sum_{r \in [k]} Q_{r-1}(S_{-i}) \cdot w_{r} \\
        &\leq (q_j - c_j) + \sum_{i \in S_{-j}} \left( q_i - \frac{q_i \cdot (1-q_j) \cdot (q_j - \varepsilon)}{q_j \cdot (1-q_i)} \right) \\
        &=  (q_j -c_j) + \sum_{i \in S_{-j}} \frac{q_i \cdot q_j \cdot (1-q_i) - q_i \cdot (1-q_j) \cdot q_j + q_i \cdot (1-q_j) \cdot \varepsilon}{q_j \cdot (1-q_i)} \\
        &=  (q_j -c_j) + \sum_{i \in S_{-j}} \frac{q_i \cdot q_j \cdot (q_j-q_i) + q_i \cdot (1-q_j) \cdot \varepsilon}{q_j \cdot (1-q_i)}  \\
        &\leq (q_j-c_j) + \sum_{i \in S_{-j}} 2 \cdot q_i \cdot q_j + \sum_{i \in S_{-j}} \frac{q_i}{q_j} \cdot \varepsilon \\
        &\leq (q_j-c_j) + 2 \cdot q_j^2 + \frac{q_j}{q_j} \cdot \varepsilon \\
        &\leq \varepsilon + 2 \cdot \varepsilon + \varepsilon,
    \end{align*}
where the second-to-last inequality follows from the fact that $\sum_{i \in S_{-j}} q_i \leq \sum_{i=1}^{j-1} q_i < q_j$ by construction of $q$, and the last inequality follows from the bound $q_j^2 \leq q_\ell ^2 = 1/2^{2\ell +2} \leq 1/2^{2\ell +1} = \varepsilon$.
    Note that the social welfare is:
    \begin{align*}
        \SW  = \sum_{i \in [\ell ]} q_i - \sum_{i \in [\ell ]} c_i = \ell \cdot \varepsilon.
    \end{align*}
    This completes the proof.
\end{proof}

\subsection{Upper Bound Governed by the Success-Probability Spread}
\label{sec:EQ_analysis}

Next we show a positive result, namely that anonymous contracts achieve an approximation guarantee that (essentially) matches the impossibility in \Cref{thm:negativegeneral}. In fact, we show that this is possible with the special class of \emph{uniform} anonymous contracts.

Namely, we show that the gap between uniform anonymous contracts and optimal contracts is of order $\min\{n,\log(nQ)\}$ where $Q$ is the maximum ratio between success probabilities. In \Cref{thm:Cbound} in \Cref{app:costs} we show an analogous result in terms of the ratios of costs.

\begin{theorem}\label{thm:Qbounds}
Let $\EQ(\cI)$ denote the principal's utility from an optimal uniform anonymous contract with limited liability on instance $\cI$.
Given any $0<a<b\leq 1$, let $$\mathcal{F}(a,b) = \left\{(\bq,\bc) \ : \ q_i \in [a,b], \ 0 \leq c_i \leq q_i \ \forall i \in [n]\right\}.$$
Denote by $Q := {b}/{a}$. 
Then we get
\[
\min\left\{\frac{n+1}{2}, 1+\frac{1}{2}\cdot \log\Big(\frac{Qn}{2\log(Qn)}\Big)\right\} \leq  \sup_{\cI \in \mathcal{F}(a,b)} \frac{\SW(\cI)}{\EQ(\cI)} \leq \min\left\{n,1+\log(Qn)\right\}.
\]
\end{theorem}

Before we provide the proof of Theorem \ref{thm:Qbounds}, we argue that uniform anonymous contracts are in fact more robust than the theorem might suggest. It is possible that a single outlier (e.g., an agent with very small success probability) arbitrarily inflates $Q$, rendering the bound weak. However, uniform anonymous contracts are not affected by such outliers provided they do not constitute the bulk of the potential welfare. We formalize this in the following corollary (whose proof is deferred to \Cref{app:proofsfromfour}), which captures a wide range of instances where most agents have reasonable cost-to-probability ratios.

\begin{restatable}{corollary}{cornondegenerate}\label{cor:non-degenerate}
Let $\EQ(\cI)$ denote the principal's utility from an optimal uniform anonymous contract with limited liability on instance $\cI$.
Let
$T = \left\{i \in [n] \ | \  {c_i}/{q_i} < 1-{1}/{n}\right\}.$
If $2\sum_{i\in T} (q_i - c_i)  >(1+\delta) \cdot \SW(\cI) $ for some constant $\delta>0$, then  
\[\SW(\cI) = \Theta(\log n) \cdot \EQ(\cI).\]
\end{restatable}

In other words, as long as those agents who do not operate on very thin margins drive the social welfare, the principal achieves a logarithmic approximation, regardless of how inefficient any outlier agents are.

If we make the stronger assumption of $c_i/q_i < \alpha$ $\forall i\in [n]$, for some constant $\alpha <1$, then it is easy to make the following, stronger conclusion (whose proof is deferred to \Cref{app:proofsfromfour}). 

\begin{restatable}{proposition}{propconstantfactorbounds}\label{prop: constant factor bounds}
Let $\mathcal{I} = \{(q_i, c_i)\}_{i \in [n]}$. If there exists a constant $\alpha \in [0, 1)$ such that $\frac{c_i}{q_i} \leq \alpha$ for all $i \in [n]$, then the utility of the optimal uniform anonymous contract, $\EQ(\cI)$, satisfies:
\[
\EQ(\cI) \geq (1 - \alpha) \cdot \SW(\cI).
\]
\end{restatable}

\subsection{Proof of Theorem \ref{thm:Qbounds}}\label{sec:EQ_worst_case}

In the remainder of this section, we give the proof of Theorem \ref{thm:Qbounds}. We recall the setting: given some interval $[a,b]\subset (0,1]$, consider any instance $\cI=(\bq,\bc)$ such that $q_i \in [a,b]$ $\forall i\in [n]$. Our main theorems (\Cref{thm:Qbounds} and \Cref{thm:Cbound}) bound the worst case multiplicative ratio of uniform anonymous contracts to social welfare. 
As an application of the upper bound, in Corollary \ref{cor:non-degenerate} we obtain an upper bound under complementary conditions, where $c_i/q_i<1-1/n$ for a significant proportion of agents.

As a first step, we derive some bounds that suggest a concrete path to proving Theorem \ref{thm:Qbounds}. We bound the ratio $\SW(\cI)/\EQ(\cI)$ by fixing the success probabilities and finding the worst-case costs. The complementary analysis fixing costs is deferred to Appendix \ref{app:costs}. To this end, Proposition \ref{prop:equal_scaled_Q} establishes a general bound for any subset of agents $S \subseteq [n]$ (whose proof is deferred to \Cref{app:proofsfromfour}). While setting $S=[n]$ yields a direct upper bound on the ratio, the flexibility to choose any subset is crucial for the proof of Corollary \ref{cor:non-degenerate}. Lemma \ref{lem:lower_h_bound} subsequently demonstrates that for any fixed probabilities, there exist costs that make the bound in Proposition \ref{prop:equal_scaled_Q} tight (whose proof is deferred to \Cref{app:proofsfromfour}).
Finally, Lemma \ref{claim:approx_sum_characterization} analyzes the function governing this tight bound to establish the final logarithmic approximation.

\begin{restatable}{proposition}{propequalscaledQ}\label{prop:equal_scaled_Q}
        Let $\cI = (\bq,\bc)$ be an instance of $n$ agents, with the ordering of indices such that the sequence $(c_i/q_i)$ is non-decreasing. Then for any $S \subseteq [n]$,
    \[
       \EqStyle \sum_{\ell\in S} (q_\ell - c_\ell) \leq   \sum_{\ell \in S} \QLEXPR \cdot \EQ(\mathcal{I}).
    \]
\end{restatable}

Furthermore, we show that for any $\cI \in \mathcal{F}(a,b)$, there exist some costs $\bc = (c_1,\ldots,c_n)$ such that the bound given in Proposition \ref{prop:equal_scaled_Q} is actually tight. More precisely, we show that if we have an instance where probabilities of success $q_1 \leq \ldots \leq q_n$ of agents are given, then for any quantity $Z\leq q_1$, there exist costs $c_1,\ldots,c_n$ such that ${c_1}/{q_1} \leq \ldots \leq {c_n}/{q_n}$, $\EQ = Z$ and the inequality in Proposition \ref{prop:equal_scaled_Q} holds with equality for all subsets $S\subseteq [n]$.

\begin{restatable}{lemma}{lemlowerhbound}\label{lem:lower_h_bound}
    Let $q_1 \leq \ldots \leq q_n$ be some set of success probabilities of agents and let $Z \leq q_1$. Then there exists $c_1,\ldots,c_n$ such that $\forall S \subseteq [n]$,
    \begin{equation}\label{eq:EQ_tight_bound}
       \EqStyle \sum_{\ell\in S} (q_\ell - c_\ell) =  \sum_{\ell \in S} \QLEXPR \cdot Z.
    \end{equation}
    Moreover, ${c_1}/{q_1} \leq \ldots \leq {c_n}/{q_n}$ and $\EQ = Z$ is achieved by any prefix subset $\{1,\ldots ,i\}$ for $i\in [n]$.
\end{restatable}

It follows that analyzing the ratio ${\SW(\cI)}/{\EQ(\cI)}$ is equivalent to analyzing the function
\[\EqStyle h(\bq) := \sum_{\ell \in [n]} \QLEXPR .\]

Our next lemma provides a nearly tight analysis of the function $h(\bq)$. Specifically, we bound the maximum value of $h(\bq)$ over an interval $[a,b]$, establishing that it scales logarithmically with respect to $n\cdot (b/a)$.
\begin{lemma}\label{claim:approx_sum_characterization}
    Consider the interval $[a,b] \subset(0,1]$ and denote the ratio $Q = {b}/{a}$. The quantity
    \begin{equation}\label{eq:max_h(q)}
       \EqStyle \max_{q_i \in [a,b] \ \forall i \in [n]} h(\bq)
    \end{equation}
    is achieved by a unique $\bq^{\star}$, with $q^{\star}_i$ non-decreasing.
    We show that
    \begin{equation}\label{eq:h(q)_bound}
     \min\left\{\frac{n+1}{2}, 1+\frac{1}{2}\cdot \log\left(\frac{Qn}{2\log(Qn)}\right)\right\} \leq  h(\bq^{\star})\leq \min\left\{n,1+\log(Qn)\right\}.
    \end{equation}
\end{lemma}
\begin{proof}
    We start by proving the upper bound on $h(\bq^\star)$. Let $F_i = \sum_{\ell \in [i]} q_\ell$, then 
    \begin{equation}\label{eq:approx_factor_equality}
        \EqStyle h(\boldsymbol{q}) = \sum_{\ell \in [n]} \QLEXPR = 1+\sum_{i \in [n-1]} \left({F_{i+1}-F_i}\right)/{F_{i+1}} = n - \sum_{i \in [n-1]} {F_i}/{F_{i+1}}.
    \end{equation}
    By the AM-GM inequality, we know that
    \begin{equation}\label{eq:am-gm}
         \EqStyle \sum_{i \in [n-1]} {F_i}/{F_{i+1}} \geq (n-1)\left(\prod_{i \in [n-1]}{F_i}/{F_{i+1}}\right)^{\frac{1}{n-1}} = (n-1)\left({F_1}/{F_n}\right)^{\frac{1}{n-1}}.
    \end{equation}
    where equality is attained when $F_i/F_{i+1}$ are equal for all $i\in [n-1]$. We know that $\forall x$, $x^{\frac{1}{n-1}}\geq 1+ \frac{1}{n-1}\log(x)$. Thus, from equation (\ref{eq:am-gm}),
    \[
    (n-1)\left({F_1}/{F_n}\right)^{\frac{1}{n-1}}
    \geq (n-1) - \log\left({F_n}/{F_1}\right).
    \]
    Substituting into (\ref{eq:approx_factor_equality}), we get
    \[
      \EqStyle \sum_{\ell \in [n]} \QLEXPR \leq 1+ \log\left({F_n}/{F_1}\right),
    \]
    The upper bound on $h(\bq)$ in (\ref{eq:h(q)_bound}) then follows simply from observing that ${F_n}/{F_1} \leq Qn$. We now proceed to find a lower bound on $h(\bq^{\star})$. A lower bound is more difficult to obtain. For a relatively tight result, we need to understand the structure of $\bq^{\star}$.

     We first show that for any fixed set $\{q_1,\ldots,q_n\}$ the permutation $\sigma$ that maximizes $h(\bq)$ is of increasing order, that is $h(\bq)\leq h(\sigma(\bq))$ where $q_{\sigma(1)}\leq \ldots \leq q_{\sigma(n)}$. This holds by a local swap argument.
     Let $F_0=0$. For any $t\in [n-1]$,
     \begin{align}
       q_t < q_{t+1}  & \Longleftrightarrow F_t < F_{t-1}+q_{t+1} \nonumber\\
         & \Longleftrightarrow  q_{t+1}\left(\frac{q_t}{(F_{t-1}+q_{t+1})F_{t+1}}\right) < q_t\left(\frac{q_{t+1}}{F_tF_{t+1}}\right) \nonumber\\
         & \Longleftrightarrow q_{t+1}\left( \frac{1}{F_{t-1}+q_{t+1}} - \frac{1}{F_{t+1}} \right) < q_t\left(\frac{1}{F_t} - \frac{1}{F_{t+1}}\right) \nonumber\\
         & \Longleftrightarrow \frac{q_{t+1}}{F_{t-1}+q_{t+1}} + \frac{q_t}{F_{t+1}} < \frac{q_t}{F_t} + \frac{q_{t+1}}{F_{t+1}}. \label{eq:swap_difference}
     \end{align}

    In other words, for any consecutive indices $t,t+1 \in [n-1]$, the difference by swapping them in $h(\bq)$ is given by the left and right hand sides of (\ref{eq:swap_difference}). This is because only summands $t$ and $t+1$ in $h(\bq)$ are affected by such a swap. We conclude that the claim on monotone increasing order of $q$ maximizing $h(\bq)$ holds.
    Hence, in order to find the solution to (\ref{eq:max_h(q)}), we may assume $q_1\leq q_2 \leq \ldots \leq q_n$. Then clearly any solution must have $q_1=a$ and $q_n=b$.
    
    We may give an explicit formulation. By (\ref{eq:am-gm}) and the unique condition for equality in the arithmetic-geometric mean inequality, we have $\frac{F_2}{F_1} = \frac{F_3}{F_2} = \frac{F_{n}}{F_{n-1}}$. Denoting this ratio $\rho = \frac{F_2}{F_1}$, we get
    \[
    \rho^{i-1} = ({F_2}/{F_1}) \cdot \ldots \cdot ({F_{i}}/{F_{i-1}}) = {F_i}/{F_1} \Longleftrightarrow F_1\rho^{i-1} = F_i \quad \forall i\in [n].
    \]
    We can express $\rho$ in terms of $Q$.
    For all $i\in [n-1]$, $F_{i+1} - F_i = q_{i+1}$. Thus
    \[
    Q = {q_n}/{q_1} = ({F_{n} - F_{n-1}})/{F_1} = ({F_1\rho^{n-1} - F_1\rho^{n-2}})/{F_1} = \rho^{n-1} - \rho^{n-2}.
    \]
    Since the function $x^{n-1}-x^{n-2}$ is monotone increasing for $x\geq 1$, there is a unique solution.
    That is, we have found $\bq^\star$. However, it has no closed form. 
    For an approximation, we lower bound $\rho$.

    Define the term $X = \frac{Qn}{2\log(Qn)}$. As we show in \Cref{lem:lemmarholowerbound}, it holds that
    $X - X^{\frac{n-2}{n-1}} < Q$. 
    Using this, we obtain that
    \begin{equation}\label{eq:rho_bound}
    \rho^{n-1} \geq Qn/(2\log(Qn)).
    \end{equation}

    To complete the proof, we next use this lower bound on $\rho$ to obtain the claimed lower bound on $h(\bq^\star)$. By \Cref{eq:approx_factor_equality} and the definition of $\rho$, we have
    \[
    h(\bq^\star) = n - \sum_{i \in [n-1]} {F_i}/{F_{i+1}} = n - (n-1)(1/\rho) = n - (n-1)\left({1}/{\rho^{n-1}}\right)^{{1}/({n-1})}.
    \]
    We want to simplify the term $(n-1)\left({1}/{\rho^{n-1}}\right)^{{1}/({n-1})}$, so we approximate it as follows. 
    By (\ref{eq:rho_bound}), we know $\rho^{n-1} \geq X$. Since the expression for $h(\bq^\star)$ is monotonically increasing with respect to $\rho$, we can substitute this lower bound to obtain:
    \[
    h(\bq^\star) \geq n - (n-1)X^{-\frac{1}{n-1}}.
    \]
    Let $y = \frac{\log X}{n-1}$. For all $Q \ge 1$ and $n \ge 2$, it is guaranteed that $X > 1$ and thus $y > 0$. We can rewrite the expression as:
    \[
    h(\bq^\star) \geq n - (n-1)e^{-y} = 1 + (n-1)(1 - e^{-y}).
    \]
    For all $y \geq 0$, the standard exponential inequality $e^y \geq 1 + y$ holds. Rearranging this yields $e^{-y} \leq \frac{1}{1+y}$, which implies $1 - e^{-y} \geq \frac{y}{1+y}$. Substituting this inequality into our expression yields:
    \[
    h(\bq^\star) \geq 1 + (n-1)\left(\frac{y}{1+y}\right).
    \]
    Substituting $y = \frac{\log X}{n-1}$ back into the equation gives:
    \[
    h(\bq^\star) \geq 1 + (n-1)\frac{\frac{\log X}{n-1}}{1 + \frac{\log X}{n-1}} = 1 + \frac{(n-1)\log X}{(n-1) + \log X}.
    \]
    Now observe that for any $a, b > 0$, the algebraic inequality $\frac{ab}{a+b} \geq  \frac{1}{2}\min\{a, b\}$ holds.
    To see this, assume without loss of generality that $a \leq b$. Then $a+b\le 2b$ and therefore $
    \frac{ab}{a+b}\ge \frac{ab}{2b}=\frac a2=\frac12\min\{a,b\}$.
    Setting $a = n-1$ and $b = \log X$, we obtain:
    \[
    h(\bq^\star) \geq 1 + \frac{1}{2}\min\left\{n-1, \log X\right\} = \min\left\{\frac{n+1}{2}, 1 + \frac{1}{2}\log X\right\}.
    \]
    Therefore, expanding $X$, we obtain the lower bound:
    \[
    h(\bq^\star) \geq \min\left\{\frac{n+1}{2}, 1 + \frac{1}{2}\log\left(\frac{Qn}{2\log(Qn)}\right)\right\}.\qedhere
    \]
    \end{proof}

    We are now ready to prove \Cref{thm:Qbounds}.

    \begin{proof}[Proof of Theorem \ref{thm:Qbounds}]
    Given $0<a,b\leq 1$, and some $\bq=(q_1,\ldots,q_n)$ with $q_i \in [a,b] \ \forall i\in [n]$. Proposition \ref{prop:equal_scaled_Q} implies that for any costs $\bc = (c_1,\ldots,c_n)$, we have ${\SW((\bq,\bc)))}/{\EQ((\bq,\bc))} \leq h(\bq)$. Now let $Z = {q_1}/{2}$. Then we know by Lemma \ref{lem:lower_h_bound} that there exists costs $c'$ with $\EQ((\bq,\bc')) = Z$ and 
    \[
    {\SW((\bq,\bc'))}/{\EQ((\bq,\bc'))} = h(\bq).
    \]
    This shows that 
    \[
    \EqStyle \sup_{\cI \in \mathcal{F}(a,b)} {\SW(\cI)}/{\EQ(\cI)} = \max_{q_i \in [a,b] \ \forall i \in [n]} h(\bq).
    \]
    Then, Lemma \ref{claim:approx_sum_characterization} proves the claim.
\end{proof}

\section{Anonymous Contracts Without Limited Liability}\label{sec:Positive Result Without Limited Liability}

In this section, we explore anonymous contracts \emph{without} the limited liability requirement.
We show that if the limited liability assumption is removed (i.e., payments to agents are allowed to be negative), 
then anonymous contracts can achieve much better approximations to the social welfare. In particular, without requiring any restrictions on the agents' probabilities of success, anonymous contracts without limited liability yield a logarithmic approximation to the social welfare. Having all agents' success probabilities distinct, \emph{even just slightly}, is enough to devise an anonymous contract extracting the whole welfare.

\subsection{General Instances}

We start by showing a positive result that applies to anonymous contracts without the limited liability requirement, showing that these achieve a logarithmic approximation to the social welfare and, thus, the optimal unrestricted contract.

\begin{theorem}[Positive Result for General Instances, Anonymous Contracts without LL]\label{thm:positive_without_ll}
For any  instance $\cI$,
let $\OPTLL(\cI)$ denote the principal's utility from an optimal anonymous contract without limited liability.
   Then it holds that:
    \begin{align*}
        \SW(\cI) = O(\log n) \cdot \OPTLL(\cI).
    \end{align*}
\end{theorem}

Notably, in \Cref{lem:lb-log} (whose proof is deferred to \Cref{sec:equal_prob}), we show the logarithmic bound in \Cref{thm:positive_without_ll} is tight. 
Interestingly, the hard instance is when all agents have identical probabilities.

\begin{restatable}{lemma}{lemlblog}\label{lem:lb-log}
    Let $q_i=q$ $\forall i\in [n]$. Then the worst case instance $\cI = (\bq,\bc)$ over all $(c_i)_{i \in [n]}$ has that $\OPT(\cI) = \EQ(\cI)$ and also has the following asymptotic behavior in function of $n$:
    \[
    \SW(\cI) = \Omega(\log(n)) \cdot \OPT(\cI).
    \]
\end{restatable}

In the remainder of this section, we discuss the proof of \Cref{thm:positive_without_ll}. Let $s_i = q_i - c_i$ for all $i \in [n]$. Without loss of generality, assume that the agents are sorted such that $s_1 \geq s_2 \geq \ldots \geq s_n$. 
If $s_i < 0$ for all $i \in [n]$, then $\SW = 0$, and the principal's utility from $\bw = (0, 0, \ldots, 0)$ is $0$. Hence, we focus on the case where $s_1 \geq 0$.
   Fix any success probabilities $q_1, \ldots, q_n$ and costs $c_1, \ldots, c_n$. Define
\[
    \EqStyle k^\star \in \argmax_{k \in [n]} \, k \cdot s_k,
\]
and consider an anonymous contract $\bw = (w_1, \ldots, w_n)$ given by:
\[
    w_j = \begin{cases}
    1 & \text{if } j < k^\star,\\
    1 - s_{k^\star} / \prod_{j \in [k^\star]} q_j & \text{if } j = k^\star,\\
    -\infty & \text{if } j > k^\star.
    \end{cases}
\]

Intuitively, this contract has two components: (1) the principal pays a reward of $1$ to each agent in $[k^\star]$ upon success, and (2) if all agents in $[k^\star]$ succeed, then each agent in $[k^\star]$ pays a transfer of $s_{k^\star} / \prod_{j \in [k^\star]} q_j$ back to the principal. Note that $s_{k^\star} / \prod_{j \in [k^\star]} q_j$ may exceed $1$, and thus the contract violates limited liability.
In \Cref{app:proofsfromfive}, we show that $S = [k^\star]$ is a Nash equilibrium.

\begin{restatable}{lemma}{prefixequilibrium}\label{lem:prefixequilibrium}
The set $S = [k^\star]$ forms a Nash equilibrium under contract $\bw$.
\end{restatable}

We are now ready to prove \Cref{thm:positive_without_ll}.
The idea behind the proof is similar to the analysis used for auctions of digital goods with unlimited supply \cite[e.g.,][]{DBLP:journals/geb/GoldbergHKSW06}.

\begin{proof}[Proof of \Cref{thm:positive_without_ll}]
We calculate the principal's expected utility under  $S = [k^\star]$, which is an equilibrium due to \Cref{lem:prefixequilibrium},
\begin{align*}
    u_p(S, \bw) &= \EqStyle \sum_{i \in [k^\star]} q_i - \sum_{j \in [k^\star]} j \cdot Q_j(S) \cdot w_j \\
    &= \EqStyle \sum_{i \in [k^\star]} q_i - \sum_{i \in [k^\star]} q_i + \left( \prod_{j \in [k^\star]} q_j \right) \cdot \left( \frac{{k^\star \cdot s_{k^\star}} }{{\prod_{j \in [k^\star]} q_j} }\right) \\
    &=  k^\star \cdot s_{k^\star}.
\end{align*}

By construction of $k^\star$, we have $k^\star \cdot s_{k^\star} \geq i \cdot s_i$ for all $i \in [n]$. Moreover, since $s_1 \geq 0$, we have that $k^\star \cdot s_{k^\star} \geq 0$, which implies that $k^\star \cdot s_{k^\star} \geq i \cdot \max(s_i,0)$ for all $i \in [n]$.  Therefore:
\[
   \EqStyle \SW = \sum_{i \in [n]} \max(s_i,0) = \sum_{i \in [n]} ({i \cdot \max(s_i,0)})/i \leq \sum_{i \in [n]} 
    ({k^\star \cdot s_{k^\star}})/{i}
    = u_p(S, \bw) \cdot \sum_{i \in [n]} 1/i .
\]
It follows that $\SW = \Theta(\log n) \cdot u_p(S, \bw)$, which completes the proof.
\end{proof}

\subsection{Distinct Success Probabilities}

Finally, we show that, under the mild assumption that the success probabilities $q_1,\dots,q_n$ are distinct, the entire social welfare can be extracted by an anonymous contract.
By potentially removing some agents from the instance, we may assume without loss of generality that $q_i > c_i$ for all $i \in [n]$, and we consider the optimization problem over anonymous contracts that incentivize $S = [n]$ as an equilibrium.\footnote{To ensure that the excluded agents do not want to exert effort, we can simply pad the contract in the reduced instance by setting $w_j = -\infty$ for all $j > n'$ where $n' \leq n$ is the number of agents in the reduced instance, similar to the contract used in the proof of \Cref{thm:positive_without_ll}.} Note that all incentive compatibility constraints are of type \eqref{eq:principal-program2}, and none are of type \eqref{eq:principal-program3}. Specifically, we have:
\begin{align}\label{eq:ic-for-all}
    \EqStyle \max_{\bw \in \mathbb{R}^n} \bone^\top(\bq - \bQ \bw) \quad \text{s.t.} \quad \bQ \bw \ge \bc, 
\end{align}
where the $(i,j)$-entry of matrix $\bQ$ is defined as $\bQ_{i,j}=q_i \cdot Q_{j-1}(S_{-i})$, and vectors $\bw, \bc$ are respectively payment and cost vectors. We can now state the following theorem:
\begin{theorem}\label{thm:invertible-noLL}
    For any  instance $\cI = (\bq,\bc)$ with $q_1, \ldots, q_n$ distinct,
    it holds that the largest principal's utility from an anonymous contract without limited liability is attained for $\bw =\bQ^{-1}\bc$ and gives
    \begin{align*}
        u_P(S,\bw) = \SW(\cI).
    \end{align*}    
\end{theorem} 

In the next lemma, we first show that, when $q_i$'s are all distinct, $\bQ$ is invertible.
\begin{lemma}\label{lem:invertibility}
Let $q_1,\dots,q_n$ be distinct positive real numbers. Then, the matrix $\bQ$ is invertible.
\end{lemma}

\begin{proof}
    For $i \in [n]$, define the polynomial $\EqStyle P_i(x) \coloneq \EqStyle q_i \cdot \prod_{k \in [n] \setminus \{i\}} \left( q_k x + (1-q_k) \right)$. Note that:
    \begin{align}
       \EqStyle P_i(x) &= \EqStyle q_i \cdot \prod_{k \in [n] \setminus \{i\}} \left( q_k x + (1-q_k) \right) \nonumber \\
       &=  \EqStyle  \sum_{S \subseteq [n] \setminus \{i\}} q_i \cdot \prod_{k \in S} (q_{k}  x) \prod_{k \in ([n] \setminus \{i\}) \setminus S} (1-q_{k})  \nonumber \\ 
       &=  \EqStyle  \sum_{j=0}^{n-1} q_i \cdot \sum_{S \subseteq [n] \setminus \{i\},  \; |S| = j}  \left( \prod_{k \in S} q_{k} \prod_{k \in ([n] \setminus \{i\}) \setminus S} (1-q_{k}) \right) \cdot x^j \nonumber \\
       &=  \EqStyle \sum_{j=0}^{n-1} q_i \cdot Q_j(S_{-i}) \cdot x^j \nonumber \\
       &= \EqStyle \sum_{j=0}^{n-1} \boldsymbol{Q}_{i,j+1} \cdot x^j \label{eq:rows_x}.
    \intertext{For each $i \in [n]$, set $z_i = 1-1/q_i$, which is well-defined since $q_i > 0$. Because $q_1, \ldots, q_n$ are pairwise distinct, so are $z_1, \ldots, z_n$.
Consider the standard Lagrange basis for polynomials of degree at most $n-1$, denoted $\ell_1(x), \ldots, \ell_n(x)$, where $\ell_i(x) = \prod_{k \in [n] \setminus \{i\}} (x-z_k) / (z_i - z_k)$ for $i \in [n]$. Then:}
        P_i(x) &= \EqStyle q_i \cdot \prod_{k \in [n] \setminus \{i\}} \left( q_k x + (1-q_k) \right) \nonumber \\
        &= \EqStyle q_i \cdot \prod_{k \in [n] \setminus \{i\}} \left( q_{k}  \left( x - (1 - 1/q_k) \right) \right)  \nonumber \\
        &=  \EqStyle q_i \cdot \prod_{k \in [n] \setminus \{i\}} \left( q_{k}  \left( x - z_{k} \right) \right) \nonumber \\
        &= \EqStyle \prod_{k \in [n]} q_k \cdot \prod_{k \in [n] \setminus \{i\}} (z_i - z_k) \cdot  \prod_{k \in [n] \setminus \{i\}} \left( \left( x - z_{k} \right) / \left(z_i - z_k \right) \right) \nonumber \\
         &= \EqStyle \prod_{k \in [n]} q_k \cdot \prod_{k \in [n] \setminus \{i\}} (z_i - z_k) \cdot  \ell_i(x). \nonumber
    \end{align}
    Hence $P_i(x)$ is a non-zero scalar multiple of $\ell_i(x)$. Since $\ell_1(x), \ldots, \ell_n(x)$ are linearly independent (e.g., \citep[Chapter 1]{phillips2003interpolation}), it follows that $P_1(x), \ldots, P_n(x)$ are also linearly independent. Writing  $\boldsymbol{Q}_i$ for the $i$-th row of $\boldsymbol{Q}$, \Cref{eq:rows_x} gives  $P_i(x) = \boldsymbol{Q}_i \cdot (1, x, x^2, \ldots, x^{n-1})$, and hence $\boldsymbol{Q}_1, \ldots, \boldsymbol{Q}_n$ are linearly independent, implying that $\boldsymbol{Q}$ is invertible.
\end{proof}

With this fact at hand, we are ready to prove \Cref{thm:invertible-noLL}:
\begin{proof}[Proof of \Cref{thm:invertible-noLL}]
    Since, by \Cref{lem:invertibility}, $\bQ$ is invertible, every $\bw \in \R^n$ can be written as $\bw = \bQ^{-1}\bu$ for some other $\bu \in \R^n$. Hence, the constraint reads $\bu \ge \bc$, and the principal's objective $\bone^\top(\bq - \bu)$. The latter is maximized when $\bu = \bc$,  that is the constraint in \eqref{eq:ic-for-all} becomes $\bw = \bQ^{-1}\bc$ (where some entries of $\bw$ may be negative). Therefore, the objective is
    \[
        \EqStyle \bone^\top(\bq - \bQ\bQ^{-1}\bc) = \sum_{i \in [n]} (q_i - c_i) = \SW,
    \]
    which concludes the proof.
\end{proof}

\clearpage

\bibliographystyle{alpha}
\bibliography{references}

\clearpage

{\noindent \LARGE  \textbf{Appendix}}

\appendix

\section{Omitted Content From Section~\ref{sec:model}: General Contracts}\label{app:model}

In this appendix, we restate and prove the main proposition on general contract:

\propunrestricted*

\begin{proof}
Fix a target set $S\subseteq[n]$ and consider the action profile where agents in $S$ exert effort and agents in $[n]\setminus S$ do not. Recall that if agent $i$ exerts effort then $\Pr[\omega_i=1]=q_i$, and if they do not exert then $\omega_i=0$. Moreover, in what follows, recall that payments satisfy limited liability, i.e., $t_i(\omega)\ge 0$.

\medskip
\noindent (1) 
If $q_i=0$, then the distribution of $\omega$ is the same whether $i$ exerts effort or not. Hence, for every $a_{-i}$,
$\E{t_i(\omega)\mid a_i=1,a_{-i}} = \E{t_i(\omega)\mid a_i=0,a_{-i}}$.
Thus the incentive constraint for $i$ reduces to $-c_i\ge 0$, which is impossible when $c_i>0$.
So no contract can incentivize a set $S$ containing an agent with $q_i=0$ and $c_i>0$.
Conversely, if $S$ contains no such agent, the contract defined in (2) below makes
the profile ``exert effort if and only if $i\in S$'' an equilibrium.\\

\noindent (2)
Fix any general contract $t_i:\{0,1\}^n\to\R_{\ge 0}$. For each agent $i$, define
\[
\tilde t_i(\omega)\;:=\; t_i(\omega)\cdot \mathbf{1}\{\omega_i=1\},\qquad
\tilde t_j:=t_j\ \ (j\neq i).
\]
Under any action profile, $\tilde t_i(\omega)\le t_i(\omega)$ pointwise, so the principal's expected
payment weakly decreases; moreover when $a_i=0$ we have $\omega_i=0$, hence
$\E{\tilde t_i(\omega)\mid a_i=0,a_{-i}}=0$.
Therefore, to incentivize exertion of effort for $i$ it is optimal to concentrate all payments on outcomes with
$\omega_i=1$.

Now suppose $i\in S$ and $q_i>0$, and $t_i(\omega)>0$ only if $\omega_i=1$.
Then
\[
\E{t_i(\omega)\mid a_i=1,a_{-i}} =  q_i\cdot \E{t_i(\omega)\mid \omega_i=1,a_i=1,a_{-i}},
\qquad
\E{t_i(\omega)\mid a_i=0,a_{-i}}=0.
\]
The incentive constraint $ \E{t_i(\omega)\mid a_i=1,a_{-i}} - c_i \ge 0$ becomes
$q_i\cdot \E{t_i(\omega)\mid \omega_i=1,\cdot}\ge c_i$.
By limited liability, the minimum expected payment satisfying this is achieved by setting
$t_i(\omega)=c_i/q_i$ whenever $\omega_i=1$ and $0$ otherwise, i.e.
\[
t_i(T)=\mathbf{1}\{i\in T\}\cdot \frac{c_i}{q_i}\qquad (i\in S,\ q_i>0),
\]
and $t_i(\cdot)=0$ for $i\notin S$ (and also for $i\in S$ with $q_i=0$, which forces $c_i=0$ by (1)).
This proves the claimed optimal form.\\

\noindent (3)
Under the above contract and the induced profile, for each $i\in S$ with $q_i>0$,
$\E{t_i(\omega)}=q_i\cdot \tfrac{c_i}{q_i}=c_i$, and for all other $i$ the expected payment is $0$.
Hence, total expected payment is $\sum_{i\in S}c_i$, and the principal's expected value is $f(S)$.
Therefore the principal's expected utility is $f(S)-\sum_{i\in S}c_i$.
\end{proof}

\section{Omitted Content From Section~\ref{sec:insights}: Existence of PNE}\label{app:equilibria}

For convenience, we restate the main theorem of Section~\ref{sec:pne} we seek to show:
\thmequilibrium*

\subsection{Proof of Lemma~\ref{lem:pne-cycle}: Equivalence between PNE and FIP}\label{sec:equiv-fip-pne}

As mentioned, we first argue that showing Theorem~\ref{thm:equilibrium} 
is implied by showing the finite improvement property, i.e., that no best-response sequence can contain a cycle. 
\lempnecyle*
\begin{proof}
    Assume the finite improvement property holds. This implies that the graph of best-response dynamics is acyclic. Since the number of possible sets $S$ that can be incentivized is finite (at most $2^n$), every best-response sequence must be finite and terminate at a node (set) $S^\star$ with no outgoing edges.
    
    At this terminal node $S^\star$, no agent can profitably deviate. Specifically, no agent inside $S^\star$ wishes to leave (preventing decrementing deviations), and no agent outside $S^\star$ wishes to join (preventing incrementing deviations). Consequently, the constraints in~\eqref{eq:principal-program2} and in~\eqref{eq:principal-program3} are satisfied. By definition, $S^\star$ is a pure Nash equilibrium implementing $\bw$.
\end{proof}

\subsection{Proof of Lemma~\ref{lem:nocycle}: Establishing the Finite Improvement Property}\label{sec:fip}

We now turn to proving Lemma~\ref{lem:nocycle}:
\lemnocyle*

\clcostcycle*
\begin{proof}
    Let $S$ be the set from which we start the walk along the cycle and let us take an arbitrary agent $i \in S$. We analyze his possible deviations: In particular, we argue that if an agent causes an incrementing deviation then his next deviation must be decrementing (but it cannot happen immediately after), and vice versa.
    
    It is clear that agent $i$ cannot cause two successive (at times $t, t^\prime$ with $t < t^\prime$) incrementing or decrementing deviations, as he cannot respectively join a set he already belongs to and cannot leave a set he does not belong to. Only a decrementing deviation can follow an incrementing one and vice versa. Indeed, such a deviation has to exist as otherwise we would end up at a different set $T \neq S$, a contradiction. Moreover, if agent $i$ causes an incrementing deviation to leave $S$ at time $t$, it cannot cause a decrementing deviation to join $S \setminus \{i\}$ at time $t+1$, as otherwise he would not have joined $S$ in the first place. 
    
    Therefore, since $S, i$ have been taken arbitrarily, it must be the case that, for all agents contained in any of the sets of the cycle, their deviations have to alternate each other, and, thus, the number of their incrementing deviations is equal to the number of their decrementing deviations. This implies, in particular, that, in a cycle, the sum of costs on all incrementing deviations is equal to the sum of costs on all decrementing ones, as desired.
\end{proof}

\clrewardreduction*
    \begin{proof}
    Consider a set $S \subseteq [n]$ with $|S|=k+1$. For the incrementing deviation caused by $i^+$, we decompose the sum based on whether $i^-$ is active (in $T$) or inactive (not in $T$). Note that if $i^- \notin T$, it contributes a factor of $(1-q_{i^-})$ to the probability product:
    \allowdisplaybreaks
    \begin{align*}
        g_{i^+}(S) &= q_{i^+} \cdot \sum_{j \in [k+1]} Q_{j-1}(S \setminus \{i^+\}) \cdot w_j \\
        &= q_{i^+} \cdot \sum_{j \in [k+1]} w_j \cdot \sum_{\substack{T \subseteq S \setminus \{i^+\}\\ |T|=j-1}} \prod_{i \in T} q_i \cdot \prod_{i \in S \setminus (T \cup \{i^+\})} (1-q_i) \\
        &= q_{i^+} \cdot \sum_{j \in [k+1]} w_j \cdot \Bigg(q_{i^-} \cdot \sum_{\substack{T \subseteq S \setminus \{i^+\}:~ i^- \in T\\ |T|=j-1}} \prod_{i \in T \setminus \{i^-\}} q_i \cdot \prod_{i \in S \setminus (T \cup \{i^+\})} (1-q_i) \\
        &\hspace{3cm} + (1-q_{i^-}) \cdot \sum_{\substack{T \subseteq S \setminus \{i^+\}:~ i^- \notin T\\ |T|=j-1}} \prod_{i \in T} q_i \cdot \prod_{i \in S \setminus (T \cup \{i^+, i^-\})} (1-q_i)\Bigg) \\
        &= (1-q_{i^-}) \cdot \underbrace{q_{i^+} \cdot \sum_{j \in [k]} w_j \cdot \sum_{\substack{T \subseteq S \setminus \{i^+\}:~ i^- \notin T\\ |T|=j-1}} \prod_{i \in T} q_i \cdot \prod_{i \in S \setminus (T \cup \{i^+, i^-\})} (1-q_i)}_{=: g^{\prime}_{i^+}(S)}\\
        &\quad + \underbrace{q_{i^+}q_{i^-} \cdot \sum_{j \in [k+1]} w_j \cdot \sum_{\substack{T \subseteq S \setminus \{i^+\}:~ i^- \in T\\ |T|=j-1}} \prod_{i \in T \setminus \{i^-\}} q_i \cdot \prod_{i \in S \setminus (T \cup \{i^+\})} (1-q_i)}_{=: g^{\prime\prime}_{i^+}(S)}.
    \end{align*}
    Observe that in the step from the third to the fourth line above, when $i^- \notin T$, the summation index changes from $j \in [k+1]$ to $j \in [k]$: this is because $j = k+1$ requires $|T|=k$, but since $T \subseteq S \setminus \{i^+\}$ and $i^- \notin T$, this implies that $T \subseteq S \setminus \{i^+, i^-\}$, that is $|T|\le |S|-2=k-1$, which is a contradiction. Hence, no such set $T$ can exist and the term for $j = k+1$ becomes zero. 
    
    By an identical derivation for $\ell_{i^-}(S)$ (decomposing based on $i^+$ and noting the factor $(1-q_{i^+})$ when $i^+ \notin T$), we obtain:
    \begin{align*}
        \ell_{i^-}(S) &= (1-q_{i^+}) \cdot \underbrace{q_{i^-} \cdot \sum_{j \in [k]} w_j \cdot \sum_{\substack{T \subseteq S \setminus \{i^-\}:~ i^+ \notin T\\ |T|=j-1}} \prod_{i \in T} q_i \cdot \prod_{i \in S \setminus (T \cup \{i^+, i^-\})} (1-q_i)}_{=: \ell^{\prime}_{i^-}(S)}\\
        &\quad + \underbrace{q_{i^+}q_{i^-} \cdot \sum_{j \in [k+1]} w_j \cdot \sum_{\substack{T \subseteq S \setminus \{i^-\}:~ i^+ \in T\\ |T|=j-1}} \prod_{i \in T \setminus \{i^+\}} q_i \cdot \prod_{i \in S \setminus (T \cup \{i^-\})} (1-q_i)}_{=: \ell^{\prime\prime}_{i^-}(S)}.
    \end{align*}
    We observe three equalities:
    \begin{itemize}
        \item[(a)] By definition, $g^{\prime}_{i^+}(S) = g_{i^+}(S \setminus \{i^-\})$ and $\ell^{\prime}_{i^-}(S) = \ell_{i^-}(S \setminus \{i^+\})$.
        \item[(b)] Expanding the definitions of $q_{i^-} g^{\prime}_{i^+}(S), q_{i^+} \ell^{\prime}_{i^-}(S)$, we observe they are equal:
        \[
            q_{i^-} g^{\prime}_{i^+}(S) = q_{i^-} q_{i^+} \sum_{j} w_j Q_{j-1}(S \setminus \{i^+, i^-\}) = q_{i^+} \ell^{\prime}_{i^-}(S).
        \]
        \item[(c)] We have that $g^{\prime\prime}_{i^+}(S)$ and $\ell^{\prime\prime}_{i^-}(S)$ are identical by reducing both to a summation over subsets of $S' = S \setminus \{i^+, i^-\}$. Recall the definition of $g^{\prime\prime}_{i^+}(S)$:
        \[
            g^{\prime\prime}_{i^+}(S) = q_{i^+}q_{i^-} \sum_{j \in [k+1]} w_j \sum_{\substack{T \subseteq S \setminus \{i^+\}:~ i^- \in T\\ |T|=j-1}} \prod_{i \in T \setminus \{i^-\}} q_i \cdot \prod_{i \in S \setminus (T \cup \{i^+\})} (1-q_i).
        \]
        We apply the change of variable $R = T \setminus \{i^-\}$. Since $T \subseteq S \setminus \{i^+\}$ and contains $i^-$, it follows that $R \subseteq S \setminus \{i^+, i^-\}$. Furthermore, $|R| = |T|-1 = j-2$, and the complement set becomes $S \setminus (T \cup \{i^+\}) = S \setminus (R \cup \{i^-, i^+\})$. Substituting this into the equation:
        \begin{equation}
            g^{\prime\prime}_{i^+}(S) = q_{i^+}q_{i^-} \sum_{j \in [k+1]} w_j \sum_{\substack{R \subseteq S \setminus \{i^+, i^-\}\\ |R|=j-2}} \prod_{i \in R} q_i \cdot \prod_{i \in S \setminus (R \cup \{i^+, i^-\})} (1-q_i). \label{eq:common-form}
        \end{equation}
        Now consider $\ell^{\prime\prime}_{i^-}(S)$:
        \[
            \ell^{\prime\prime}_{i^-}(S) = q_{i^+}q_{i^-} \sum_{j \in [k+1]} w_j \sum_{\substack{T \subseteq S \setminus \{i^-\}:~ i^+ \in T\\ |T|=j-1}} \prod_{i \in T \setminus \{i^+\}} q_i \cdot \prod_{i \in S \setminus (T \cup \{i^-\})} (1-q_i).
        \]
        We apply the change of variable $R = T \setminus \{i^+\}$. Since $T \subseteq S \setminus \{i^-\}$ and contains $i^+$, we again have $R \subseteq S \setminus \{i^-, i^+\}$. With $|R|=j-2$ and the complement $S \setminus (T \cup \{i^-\}) = S \setminus (R \cup \{i^+, i^-\})$, this term reduces to precisely the same expression as \eqref{eq:common-form}. Thus,
        \[
            g^{\prime\prime}_{i^+}(S) = \ell^{\prime\prime}_{i^-}(S).
        \]
    \end{itemize}
    Finally, we compute the difference:
    \begin{align*}
        g_{i^+}(S) - \ell_{i^-}(S) &= \left( (1-q_{i^-}) g^{\prime}_{i^+}(S) + g^{\prime\prime}_{i^+}(S) \right) - \left( (1-q_{i^+}) \ell^{\prime}_{i^-}(S) + \ell^{\prime\prime}_{i^-}(S) \right) \\
        &= g^{\prime}_{i^+}(S) - \ell^{\prime}_{i^-}(S) + \underbrace{(g^{\prime\prime}_{i^+}(S) - \ell^{\prime\prime}_{i^-}(S))}_{=0 \text{ by (c)}} + \underbrace{(q_{i^+} \ell^{\prime}_{i^-}(S) - q_{i^-} g^{\prime}_{i^+}(S))}_{=0 \text{ by (b)}} \\
        &\stackrel{\text{(a)}}{=} g_{i^+}(S \setminus \{i^-\}) - \ell_{i^-}(S \setminus \{i^+\}).
    \end{align*}
    The claim follows.
\end{proof}

\clpeakwalk*
\begin{proof}
    We start by the agent that first leaves after the peak, calling this deviation $t^-_1$. By Proposition~\ref{cl:reward-reduction}, we know that, for all $t < t^-_1$ in $P_r$, after the removal of $i(t^-_1)$,
    \begin{align*}
        &\sum_{t \in \ID(P_r)} g_{i(t)}(S(t)) - \ell_{i(t^-_1)}(S(t^-_1)) \\
        = &\sum_{t \in \ID(P_r)} g_{i(t)}(S(t) \setminus \{i(t^-_1)\}) - \ell_{i(t^-_1)}\left(S(t^-_1) \setminus \bigcup_{t^\prime \in \ID(P_r)} \{i(t^\prime)\}\right).
    \end{align*}
    unless there exists a $t \in \ID(P_r)$ such that $i(t) = i(t^-_1)$ (i.e., $i(t^-_1)$ previously caused an incrementing deviation within $P_r$). In that case, denote such incrementing deviation as $t^+_1$, so that $i(t^+_1)=i(t^-_1)$. Then, we recognize that 
    \[
        S(t^+_1) = S(t^-_1) \setminus \bigcup_{t \in \ID(P_r): t > t^+_1} \{i(t)\},
    \]
    and, thus,
    \[
        g_{i(t^+_1)}(S(t^+_1)) - \ell_{i(t^-_1)}\left(S(t^-_1) \setminus \bigcup_{t \in \ID(P_r): t > t^+_1} \{i(t)\}\right) = 0.
    \]
    Therefore, in this case, the previous expression simplifies to
    \begin{align*}
        &\sum_{t \in \ID(P_r)} g_{i(t)}(S(t)) - \ell_{i(t^-_1)}(S(t^-_1)) \\
        = &\sum_{t \in \ID(P_r)} g_{i(t)}(S(t) \setminus i(t^-_1)) - \ell_{i(t^-_1)}\left(S(t^-_1) \setminus \bigcup_{t^\prime \in \ID(P_r)} \{i(t^\prime)\}\right)\\
        = &\sum_{t \in \ID(P_r): i(t) \neq i(t^-_1)} g_{i(t)}(S(t) \setminus \{i(t^-_1)\}).
    \end{align*}
    Now, consider the second agent that leaves after the peak, calling this deviation $t^-_2$. Again, by Proposition~\ref{cl:reward-reduction}, we know that, for all $t < t^-_2$ in $P_r$, after the removal of $i(t^-_1), i(t^-_2)$,
    \begin{align*}
        &\sum_{t \in \ID(P_r)} g_{i(t)}(S(t) \setminus i(t^-_1)) - \ell_{i(t^-_2)}(S(t^-_2)) \\
        = &\sum_{t \in \ID(P_r)} g_{i(t)}(S(t) \setminus \{i(t^-_1), i(t^-_2)\}) - \ell_{i(t^-_2)}\left(S(t^-_2) \setminus \bigcup_{t^\prime \in \ID(P_r)} \{i(t^\prime)\}\right),
    \end{align*}
    unless there exists a $t \in \ID(P_r)$ such that $i(t) = i(t^-_2)$ (i.e., $i(t^-_2)$ previously caused an incrementing deviation within $P_r$). In that case, the argument is identical to that of agent $i(t^-_1)$. The claim follows by induction.
\end{proof}

We finally proceed with the proof of Lemma~\ref{lem:nocycle}.

\begin{proof}[Proof of Lemma~\ref{lem:nocycle}]
    Fix $\bw \in \R^n$. For the sake of contradiction, assume that there is a best-response sequence forming a cycle. 
    First, by Claim~\ref{cl:cost-cycle}, we have that $C := \sum_{t \in \ID} c_{i(t)} = \sum_{t \in \DD} c_{i(t)}$.
    Hence,
    \[
        \sum_{t \in \ID} g_{i(t)}(S(t)) > C > \sum_{t \in \DD} \ell_{i(t)}(S(t)).
    \]
    In turn, this implies that 
    \begin{align*}
        \sum_{t \in \ID} g_{i(t)}(S(t)) - \sum_{t \in \DD} \ell_{i(t)}(S(t)) > 0.
    \end{align*}
    In the remainder of the proof, we show that, in fact, $\sum_{t \in \ID} g_{i(t)}(S(t)) - \sum_{t \in \DD} \ell_{i(t)}(S(t)) = 0$, yielding a contradiction. We do so by considering an arbitrary agent that joins (causes an incrementing deviation) in peak walk $P_a$ at time $t^+_a$ and leaves (causes a decrementing deviation) in peak walk $P_b$ at a later time $t^-_b$. Note that $i(t^+_a)=i(t^-_b)=:i^\star$ by definition. Given that we lie on a cycle, we know that, to every incrementing deviation of $i^\star$, there corresponds a later decrementing deviation of his (before any other possible incrementing deviation). This means that agent $i^\star$ causes an incrementing deviation at $t^+_a$, a decrementing one at $t^-_b$, and no other deviation in between. 
    
    First, let us rewrite the left hand side of the earlier inequality as a sum across peak walks: 
    \begin{align}\label{eq:rews-cost-cycle}
        \sum_{r \in [m]} \left(\sum_{t \in \ID(P_r)} g_{i(t)}(S(t)) - \sum_{t \in \DD(P_r)} \ell_{i(t)}(S(t))\right) > 0.
    \end{align}
    By Claim~\ref{cl:peak-walk}, we have that, for all $r \in [m]$,
    \begin{align*}
        &\sum_{t \in \ID(P_r)} g_{i(t)}(S(t)) - \sum_{t \in \DD(P_r)} \ell_{i(t)}(S(t)) \\
        = &\sum_{t \in \ID^\star(P_r)} g_{i(t)}\left(S(t) \setminus \bigcup_{t^\prime \in \DD(P_r)} \{i(t^\prime)\}\right) - \sum_{t \in \DD^\star(P_r)} \ell_{i(t)}\left(S(t) \setminus \bigcup_{t^\prime \in \ID(P_r)} \{i(t^\prime)\}\right).
     \end{align*}
    Hence, we do not need to take into account agents that join and leave within the same peak walk. We refer to these sums over $\ID^\star(P_r)$ and $\DD^\star(P_r)$ as \emph{irreducible-in-$P_r$}. Considering just irreducible-in-$P_r$ sums is not only without loss (as per Claim~\ref{cl:peak-walk}), but it also means that we can assume that $P_a \neq P_b$, for the agent $i^\star$ we are considering.
    
    We now prove that
    \begin{align}\label{eq:agent-cancel}
        g_{i(t^+_a)}\left(S(t^+_a) \setminus \bigcup_{r \in \{a, \ldots, b\}}\bigcup_{t^\prime \in \DD\left( P_r\right): t^\prime < t^-_b} \{i(t^\prime)\}\right) - \ell_{i(t^-_b)}\left(S(t^-_b) \setminus \bigcup_{r \in \{a, \ldots, b\}}\bigcup_{t^\prime \in \ID\left(P_r\right): t^\prime > t^+_a} \{i(t^\prime)\}\right) = 0.
    \end{align}
    Recall that $i(t^+_a)=i(t^-_b)=i^\star$ is taken arbitrarily, and, hence, the conclusion in~\eqref{eq:agent-cancel} holds for all agents in the cycle. Therefore,
    \[
        0 = \sum_{r \in [m]} \left(\sum_{t \in \ID(P_r)} g_{i(t)}(S(t)) - \sum_{t \in \DD(P_r)} \ell_{i(t)}(S(t))\right) > 0.
    \]
    a contradiction, and the lemma follows. 
    
    Thus, we need to show that~\eqref{eq:agent-cancel} holds. To do so, consider two consecutive walks $P_r, P_{r+1}$ with $a \leq r, r+1 \leq b$ in their respective irreducible-in-$P_r$ and irreducible-in-$P_{r+1}$ form. Let $t^+_1(P_r)$ be the time corresponding to the last agent joining in peak walk $P_r$ and let $t^-_1(P_{r+1})$ be the time corresponding to first agent leaving in peak walk $P_{r+1}$. In their irreducible forms, we have that their reward terms have become 
    \[
       g_{i(t^+_1(P_r))}\left(S(t^+_1(P_r)) \setminus \bigcup_{t \in \DD(P_r)} \{i(t)\}\right) ~\text{and}~ \ell_{i(t^-_1(P_{r+1}))}\left(S(t^-_1(P_{r+1})) \setminus \bigcup_{t \in \ID(P_{r+1})} \{i(t)\}\right).
    \]
    The first appears with a plus sign and the second with a minus sign, since one is joining and the other leaving. 
    
    We observe that $S(t^+_1(P_r)) \setminus \bigcup_{t \in \DD(P_r)} i(t)$ is the set of agents that joined before $i(t^+_1(P_r))$ did together $i(t^+_1(P_r))$ himself, removing all agents that left by the end of peak walk $P_r$. Similarly, $S(t^-_1(P_{r+1})) \setminus \bigcup_{t \in \ID(P_{r+1})} i(t)$ is the set of agents leaving after $i(t^-_1(P_{r+1}))$ does together with $i(t^-_1(P_{r+1}))$ himself, removing all agents that joined by the end of peak walk $P_{r+1}$. Therefore,
    \[
        S(t^+_1(P_r)) \setminus \bigcup_{t \in \DD(P_r)} \{i(t)\} = S(t^-_1(P_{r+1})) \setminus \bigcup_{t \in \ID(P_{r+1})} \{i(t)\},
    \]
    as the set on the right hand side is the set before any agent joined it in peak walk $P_{r+1}$, which is exactly the same as set $S(t^+_1(P_r))$ deprived of all agents that left before the start of peak walk $P_{r+1}$. Thus, the expected payments of $i(t^+_1(P_r))$ and $i(t^-_1(P_{r+1}))$ further reduce to
    \[
       g_{i(t^+_1(P_r))}\left(S(t^+_1(P_r)) \setminus \left(\bigcup_{t \in \DD(P_r)} \{i(t) \cup i(t^-_1(P_{r+1}))\}\right)\right),
    \]
    and
    \[
        \ell_{i(t^-_1(P_{r+1}))}\left(S(t^-_1(P_{r+1})) \setminus \left(\bigcup_{t \in \ID(P_{r+1})} \{i(t) \cup i(t^+_1(P_r))\}\right)\right).
    \]
    If it so happens that $i(t^+_1(P_r)) = i(t^-_1(P_{r+1}))$, then the two reduced reward expressions cancel. Otherwise, let us take agent $i(t^-_2(P_{r+1}))$, the second agent that leaves after the peak in $P_{r+1}$. Note that, in the irreducible-in-$P_{r+1}$ form, we have
    \begin{align*}
        S(t^-_2(P_{r+1})) \setminus \bigcup_{t \in \ID(P_{r+1})} \{i(t)\} &= S(t^-_1(P_{r+1})) \setminus \left(\bigcup_{t \in \ID(P_{r+1})} \{i(t) \cup i(t^-_1(P_{r+1}))\}\right)\\
        &= S(t^+_1(P_r)) \setminus \left(\bigcup_{t \in \DD(P_r)} \{i(t) \cup i(t^-_1(P_{r+1}))\}\right),
    \end{align*}
    because the former set is the set of agents that had joined before the start of peak walk $P_{r+1}$ and so is the latter. Again, if it so happens that $i(t^+_1(P_r)) = i(t^-_2(P_{r+1}))$, then the corresponding expected payments cancel out. Iterating backwards over the deviations in the peak walks $P_r, P_{r+1}$, we conclude that that the expected payments in the \emph{irreducible-in-$(P_r \cup P_{r+1})$} form read
    \begin{align*}
        g_{i(t)}\left(S(t) \setminus \bigcup_{t^\prime \in \DD(P_r, P_{r+1})} \{i(t^\prime)\}\right) ~\forall t \in \ID(P_r) \quad \text{and} \quad \ell_{i(t)}\left(S(t) \setminus \bigcup_{t^\prime \in \ID(P_r, P_{r+1})} \{i(t^\prime)\}\right) ~\forall t \in \DD(P_{r+1}).
    \end{align*}
    Similarly, in their \emph{irreducible-in-$(P_{r-1} \cup P_r)$} and in their \emph{irreducible-in-$(P_{r+1} \cup P_{r+2})$} forms, respectively, we obtain
    \begin{align*}
        &\ell_{i(t)}\left(S(t) \setminus \bigcup_{t^\prime \in \ID(P_{r-1}, P_r)} \{i(t^\prime)\}\right) ~&\forall t \in \DD(P_r) \quad \text{and} \quad g_{i(t)}\left(S(t) \setminus \bigcup_{t^\prime \in \DD(P_{r+1}, P_{r+2})} \{i(t^\prime)\}\right) ~&\forall t \in \ID(P_{r+1}).
    \end{align*}
    Hence, by inducting backwards over the peak walks from $b$ to $a$, we have that the entering and leaving expected payments of agent $i^\star=i(t^+_a)=i(t^-_b)$ in their \emph{irreducible-in-$(\cup_{\ell \in\{a, \ldots, b\}} P_r)$} form read respectively
    \begin{align*}
        g_{i(t^+_a)}\left(S(t^+_a) \setminus \bigcup_{t^\prime \in \DD(P_a,\ldots, P_b): t^\prime < t^-_b} \{i(t^\prime)\}\right) &= g_{i^\star}\left(S(t^+_a) \setminus \bigcup_{\ell \in \{a, \ldots, b\}}\bigcup_{t^\prime \in \DD\left( P_r\right): t^\prime < t^-_b} \{i(t^\prime)\}\right)\\
        \ell_{i(t^-_b)}\left(S(t^-_b) \setminus \bigcup_{t^\prime \in \ID(P_a,\ldots, P_b): t^\prime > t^+_a} \{i(t^\prime)\}\right) &= \ell_{i^\star}\left(S(t^-_b) \setminus \bigcup_{\ell \in \{a, \ldots, b\}}\bigcup_{t^\prime \in \ID\left(P_r\right): t^\prime > t^+_a} \{i(t^\prime)\}\right),
    \end{align*}
    and also that 
    \[
        S(t^+_a) \setminus \bigcup_{t^\prime \in \DD(P_a,\ldots, P_b): t^\prime < t^-_b} \{i(t^\prime)\} = S(t^-_b) \setminus \bigcup_{t^\prime \in \ID(P_a,\ldots, P_b): t^\prime > t^+_a} \{i(t^\prime)\}.
    \]
    Hence, the reward expressions for agent $i^\star$ cancel. The lemma follows.
\end{proof}

\section{Omitted Content From Section~\ref{sec:limited_liabbility}: Anonymous Contracts With Limited Liability}\label{app:proofsfromfour}

\subsection{Positive Results for Bounded Instances}\label{app:pos_res_bounded}

While the global ratio $Q$ may be unbounded, we note in the following proof that if a subset of agents has non-negligible $1/n$ margins, the ratio of success probabilities is naturally bounded by $n^2$. 
If moreover this subset of agents contributes the majority of social welfare, we obtain a $\Theta(\log n)$ approximation by applying Theorem \ref{thm:Qbounds} to this restricted instance, where the effective $Q \le n^2$.

\cornondegenerate*
\begin{proof}
    We first give the upper bound on $\SW$.
    The condition of $i \in T$ is equivalent to
    \begin{equation}\label{eq:assumption_1-1/n}
    \frac{(q_i - c_i)}{q_i} > \frac{1}{n}
    \Longleftrightarrow n(q_i - c_i) > q_i.
    \end{equation}
    Let \[S \subseteq [n] = \{i \in [n] \ | \ \max_{j\in [n]} (q_j-c_j) \leq n(q_i-c_i)\}.\]
    Then $\sum_{i \in S} (q_i - c_i) \geq \sum_{i\in [n]\setminus S}(q_i - c_i)$, implying 
    \[
    2\cdot \sum_{i \in S} (q_i - c_i) \geq \sum_{i\in [n]}(q_i - c_i) = \SW.
    \]
    By the condition on $T$ we get
    \[
    \sum_{i \in S} (q_i - c_i) + \sum_{i \in T} (q_i - c_i) \geq \left(1+ \frac{\delta}{2}\right)\SW.
    \]
    In particular, this implies $S\cap T \neq \varnothing$ and in fact
    $\sum_{i \in S \cap T}q_i-c_i \geq \delta/2 \cdot \SW$.
    
    For any $i,j \in S \cap T$, by definition of $S$ and property (\ref{eq:assumption_1-1/n}), we have
    \[
    n^2(q_j-c_j)\geq n(q_i-c_i) > q_i
    \]
    In particular, $n^2(q_j-c_j) > q_i \Rightarrow n^2> q_i/q_j$. Thus by Proposition \ref{prop:equal_scaled_Q}, up to a factor $2$, it suffices to bound
    \[
    \sum_{\ell \in S} \frac{q_\ell}{\sum_{i \in [\ell]} q_i}
    \leq \sum_{\ell \in S} \frac{q_\ell}{\sum_{i \in S, i\leq \ell} q_i}
    \]
    where for all $i,j \in S\cap T$, $q_i/q_j \leq n^2$.
    Here, we know it may be that $|S|<n$, however, we do not make use of this assumption. Hence the problem has been reduced to simply bounding $h(q)$ where we assume $q_i/q_j \leq n^2$ $\forall i,j \in [n]$. Thus applying Theorem \ref{thm:Qbounds} on $S\cap T$, knowing that $Q = n^2$, yields the result. The asymptotically matching lower bound follows from an instance we have already analyzed, in Lemma \ref{lem:last_prem_lower_bound}.
\end{proof}

\propconstantfactorbounds*
\begin{proof}
By Proposition \ref{prop:equal}, the optimal uniform anonymous contract utility is given by $\EQ(\mathcal{I}) = \max_{k \in [n]} (1 - \frac{c_k}{q_k}) \sum_{j=1}^k q_j$, assuming indices are sorted by non-decreasing cost-to-probability ratio.

Consider the specific candidate contract that incentivizes the full set of agents $S=[n]$. To satisfy the participation constraints for all agents, the uniform payment $w$ must satisfy $w \geq \max_{i \in [n]} \frac{c_i}{q_i}$. By assumption, $\max_{i \in [n]} \frac{c_i}{q_i} \leq \alpha$. The principal's utility for this specific contract provides a lower bound for the optimal $\EQ(\cI)$:
\begin{align*}
\EQ(\cI) &\geq \left(1 - \max_{i \in [n]} \frac{c_i}{q_i}\right) \sum_{i \in [n]} q_i \geq (1 - \alpha) \sum_{i \in [n]} q_i.
\end{align*}
Since $\sum_{i \in [n]} q_i \geq \SW(\cI)$, we get 
\[
\EQ(\cI) \geq (1 - \alpha) \SW(\cI). \qedhere
\]
\end{proof}

\begin{lemma}\label{lem:lemmarholowerbound}
Let $Q > 1$ and $n \ge 2$. Define the term $X = \frac{Qn}{2\log(Qn)}$. Then the following inequality holds:
    \begin{equation}
        X - X^{\frac{n-2}{n-1}} < Q.
    \end{equation}
\end{lemma}

\begin{proof}
Let LHS (left hand side) denote the expression $X - X^{\frac{n-2}{n-1}}$. We begin by factoring out $X$ and rewriting the exponent:
\begin{align*}
    \text{LHS} &= X - X^{1 - \frac{1}{n-1}} \\
    &= X \left( 1 - X^{-\frac{1}{n-1}} \right) \\
    &= X \left( 1 - \exp\left( -\frac{\log X}{n-1} \right) \right).
\end{align*}
We apply the elementary inequality $1 - e^{-y} \le y$, which holds for all $y \in \mathbb{R}$. Setting $y = \frac{\log X}{n-1}$, we obtain:
\begin{equation*}
    \text{LHS} \le X \left( \frac{\log X}{n-1} \right).
\end{equation*}
Substituting the definition $X = \frac{Qn}{2\log(Qn)}$ back into the inequality yields:
\begin{equation*}
    \text{LHS} \le \frac{Qn}{2\log(Qn)} \cdot \frac{1}{n-1} \cdot \log\left( \frac{Qn}{2\log(Qn)} \right).
\end{equation*}
Expanding the logarithm term $\log\left( \frac{Qn}{2\log(Qn)} \right) = \log(Qn) - \log(2\log(Qn))$, we have:
\begin{align*}
    \text{LHS} &\le \frac{Qn}{2(n-1)\log(Qn)} \left[ \log(Qn) - \log(2\log(Qn)) \right] \\
    &= \frac{Qn}{2(n-1)} \left( 1 - \frac{\log(2\log(Qn))}{\log(Qn)} \right).
\end{align*}
Since $Q > 1$ and $n \ge 2$, we have $Qn > 2$. Consequently, $\log(Qn) > \log 2 \approx 0.69$. It follows that $2\log(Qn) > 1$, implying $\log(2\log(Qn)) > 0$. Thus, the factor in the parentheses is strictly less than 1:
\begin{equation*}
    \text{LHS} < \frac{Qn}{2(n-1)}<Q. \qedhere
\end{equation*}
\end{proof}

\subsection{Proof of Theorem 4.3}

\propequalscaledQ*
\begin{proof}
    Let us recall from Proposition~\ref{prop:equal} that the optimal 
    utility from a uniform anonymous contract is given by index
    \[
        k = \arg \max_{\ell \in [n]} \left(1 - \frac{c_\ell}{q_\ell}\right) \cdot \sum_{i \in [\ell]} q_i. 
    \]
    That is, 
    \[
        \EQ = \left(1 - \frac{c_k}{q_k}\right) \cdot \sum_{i \in [k]} q_i \geq \left(1 - \frac{c_\ell}{q_\ell}\right) \cdot \sum_{i \in [\ell]} q_i,
    \]
    for all $\ell \in [n]$. Rearranging,
    \[
        1- \frac{c_\ell}{q_\ell} \leq \left(1 - \frac{c_k}{q_k}\right) \cdot \frac{\sum_{i \in [k]} q_i}{\sum_{i \in [\ell]} q_i}.
    \]
    We next bound the welfare in terms of \EQ:
    \[
        \sum_{\ell \in S} q_\ell \left(1-\frac{c_\ell}{q_\ell}\right) \leq \sum_{\ell \in S} q_\ell \left(1-\frac{c_k}{q_k} \right)\cdot \frac{\sum_{i \in [k]} q_i}{\sum_{i \in [\ell]} q_i} = \EQ \cdot \sum_{\ell \in S} \frac{q_\ell}{\sum_{i \in [\ell]} q_i}.\qedhere
    \]
\end{proof}

\lemlowerhbound*
\begin{proof}
    We proceed inductively. First, we set $c_1$ to be such that $q_1-c_1 = Z$. Assume inductively that $c_1,\ldots,c_\ell$ have been well defined and satisfy all conclusions except (\ref{eq:EQ_tight_bound}). We next define $c_{\ell+1}$ such that
    \[
    \left(1-\frac{c_\ell}{q_\ell}\right)\sum_{i\leq \ell}q_i   = \left(1-\frac{c_{\ell+1}}{q_{\ell+1}}\right)\sum_{i\leq \ell+1}q_i. 
    \]
    First, such value $c_{\ell+1}$ exists since the function 
   \[
 f(x) = \left(1-\frac{x}{q_{\ell+1}}\right)\sum_{i\leq \ell+1}q_i 
   \]    
   is continuous monotone decreasing from $0$ to $q_{\ell+1}$ with $f(0) =\sum_{i\leq \ell+1}q_i$ and $f(q_{\ell+1})=0$.
    Then clearly $\frac{c_\ell}{q_\ell} \leq \frac{c_{\ell+1}}{q_{\ell+1}}$. By definition, $\EQ$ is also achieved by subset $\{1,\ldots,\ell+1\}$. Finally, the fact that (\ref{eq:EQ_tight_bound}) holds follows simply from the fact that for the $c_1,\ldots,c_n$ as just defined, all inequalities in Proposition \ref{prop:equal_scaled_Q} are tight. 
\end{proof}

\subsection{Additional Results for Bounded Instances}\label{app:costs}

Similar to Section~\ref{sec:limited_liabbility}, we also get an upper bound when the ratio of entries of $c$ is bounded. 

\begin{theorem}\label{thm:Cbound}
Given any $0<a<b\leq 1$, let $\mathcal{G}(a,b) = \left\{(\bq,\bc) \ : \ c_i \in [a,b], \ 0 \leq c_i \leq q_i \ \forall i \in [n]\right\}$. Denote by $C := \frac{b}{a}$. Then we get
\[
 \sup_{\cI \in \mathcal{G}(a,b)} \frac{\SW(\cI)}{\EQ(\cI)} \leq \min\left\{n,1+\log(Cn)\right\}.
\]
\end{theorem}

\begin{proposition}\label{prop:equal_scaled_C}
    Let $\cI = (\bq,\bc)$ be an instance of $n$ agents, with the ordering of indices such that the sequence $(c_i/q_i)$ is non-decreasing. Then for any $S \subseteq [n]$, 
    \[
        \sum_{\ell \in S} (q_\ell - c_\ell) \leq   \left(\sum_{\ell \in S} \frac{c_\ell}{\sum_{i \in [\ell]} c_i} \right) \cdot \EQ.
    \]
\end{proposition}

\begin{proof}
Let $t = \argmax_{\ell \in [n]}\left(\frac{q_{\ell}}{c_{\ell}} - 1 \right) \cdot \sum_{i \in [\ell]} c_i$. Rewriting we get
\begin{align*}
     \left(\frac{q_{t}}{c_{t}} - 1 \right) \cdot \sum_{i \in [t]} c_i \geq \left(\frac{q_{\ell}}{c_{\ell}} - 1 \right) \cdot \sum_{i \in [\ell]} c_i \Longleftrightarrow  \left(\frac{q_{t}}{c_{t}} - 1 \right) \cdot \left(\frac{\sum_{i \in [t]} c_i}{\sum_{i \in [\ell]} c_i} \right) \geq \left(\frac{q_{\ell}}{c_{\ell}} - 1 \right).
\end{align*}

Thus, we get
\begin{align*}\label{eq:SW-costs}
   \sum_{\ell \in S} c_\ell \left(\frac{q_\ell}{c_\ell} -1 \right)
   \leq  \sum_{\ell \in S} c_\ell \left(\frac{q_{t}}{c_{t}} - 1 \right) \cdot \left(\frac{\sum_{i \in [t]} c_i}{\sum_{i \in [\ell]} c_i} \right) =  \left(\frac{q_{t}}{c_{t}} - 1 \right) \cdot \left(\sum_{i \in [t]} c_i\right) \cdot \sum_{\ell \in S} \frac{c_\ell}{\sum_{i \in [\ell]} c_i}.
\end{align*}

Let us remark the following: $\left(\frac{q_{t}}{c_{t}} - 1 \right) / \left(1-\frac{c_t}{q_t}\right)
 = \frac{q_t}{c_t}$. Moreover, since indices are ordered such that $\frac{c_1}{q_1} \leq \ldots \leq \frac{c_n}{q_n}$, we have that $\left(\sum_{i \in [t]} c_i\right) / \left(\sum_{i \in [t]} q_i\right) \leq \frac{c_t}{q_t}$. Hence, we get that 
\[
\left(\frac{q_{t}}{c_{t}} - 1 \right)\left(\sum_{i \in [t]} c_i\right)
\leq  \left(1-\frac{c_t}{q_t}\right)\left(\sum_{i \in [t]} q_i\right) \leq \EQ. \qedhere
\]
\end{proof}

\begin{proof}[Proof of Theorem \ref{thm:Cbound}]
Given $0<a,b\leq 1$, and some $c=(c_1,\ldots,c_n)$ with $c_i \in [a,b] \ \forall i\in [n]$. Proposition \ref{prop:equal_scaled_C} implies that for any probabilities $q = (q_1,\ldots,q_n)$, we have $\frac{\SW((q,c)))}{\EQ((q,c))} \leq h(c)$. Then the result follows from \Cref{claim:approx_sum_characterization}.
\end{proof}

\section{Omitted Content From Section~\ref{sec:Positive Result Without Limited Liability}: Anonymous Contracts Without Limited Liability}\label{app:proofsfromfive}

\prefixequilibrium*
\begin{proof}
First, note that no agent $i \notin S$ benefits from exerting effort. If they do, there is a non-zero probability that $k^\star + 1$ agents succeed, 
and the expected utility for agent $i$ becomes: 
\begin{align*}
    u_i(S \cup \{i\}, \bw) \leq Q_{k^\star+1}(S \cup \{i\}) \cdot w_{k^\star+1} = q_i \cdot \prod_{j \in S} q_j \cdot (-\infty) = -\infty.
\end{align*}

Now, consider any agent $i \in S$. Their expected utility from exerting effort is:
\begin{align*}
    u_i(S, \bw) &= q_i \cdot \sum_{j \in [k^\star]} Q_{j-1}(S_{-i}) \cdot w_j - c_i = 
    q_i \cdot 1 - \left( \prod_{j \in [k^\star]} q_j \right) \cdot \left( \frac{s_{k^\star}}{\prod_{j \in [k^\star]} q_j} \right) - c_i \\
    &= q_i - s_{k^\star} - c_i = s_i - s_{k^\star} \geq 0,
\end{align*}
where the inequality follows from the ordering $s_1 \geq \ldots \geq s_{k^\star} \geq \ldots \geq s_n$.
Since the utility of not exerting effort is $0$, each agent in $S$ is incentivized to exert effort, confirming that $S$ is a PNE.
\end{proof}

\section{Additional Results for Uniform Anonymous Contracts}\label{sec:EQSpecialCases}

In this section, we construct two worst case instances. For the first one, we let all probabilities $q_i$ be equal. We show there exists a set of costs $\bc$ so that for $\cI = (\bq,\bc)$, we have $\EQ(\cI) = \OPT(\cI)$ and $\EQ(\cI) = \Theta(\log(n))$. That is, even though optimal payments $w$ happen to be non-negative, this lower bound matches the upper bound of Theorem \ref{thm:positive_without_ll}.

\subsection{Equal Success Probabilities}\label{sec:equal_prob}
We start with the case of equal success probabilities for all agents:

\lemlblog*

\begin{proof}
    Let $q$ be a success probability and $c_1 \leq \ldots \leq c_n$ costs to be fixed later. We have that, for any set of $k$ agents, the objective function (\ref{eq:principal-program1}) becomes 
    \[
        kq - \sum_{j \in [k]} j \cdot \binom{k}{j} q^j(1-q)^{k-j} \cdot w_j.
    \]
    Similarly, since costs are increasing in their index, the first set of constraints (\ref{eq:principal-program2}) becomes 
    \begin{align*}
        &  q_i \cdot \sum_{j \in [k]} Q_{j-1}(S_{-i}) \cdot w_j \geq c_i \\
       ~\Longleftrightarrow~  &q\sum_{j \in [k]} \binom{k-1}{j-1} q^{j-1}(1-q)^{k-j} \cdot {w_j} \geq c_k \\
        ~\Longleftrightarrow~ &\sum_{j \in [k]} j \cdot \binom{k}{j} q^j(1-q)^{k-j} \cdot w_j \geq kc_k.        
    \end{align*}
    This means we can upper-bound the principal's utility by $(q - c_k) \cdot k$. Note that this can be achieved by uniform anonymous contracts: Indeed, for all $k$
    \[
        \left(1 - \frac{c_k}{q}\right) \cdot kq = (q-c_k)\cdot k.
    \]
    Now, let us fix $q$ and $c_1=c$, and choose
    \[
        c_i = \left(1-\frac{1}{i}\right) \cdot q + \frac{c}{i}~\forall i \in [n].
    \]
    Then, $\OPT = q-c$ implying $\OPT = \EQ$. On the other hand, $\SW = (q-c) \cdot \sum_{i \in [n]} \frac{1}{i}$, and so
    \[
        \frac{\SW}{\OPT} = \sum_{i \in [n]} \frac{1}{i} \geq \Omega(\log n). \qedhere
    \] 
\end{proof}

\begin{proposition}\label{prop:equal prob. SW-OPT log(n)}
 Let $q_i=q$ $\forall i\in [n]$. Then the worst case instance $\cI = (\bq,\bc)$ over all $(c_i)_{i \in [n]}$ has the following asymptotic behavior in function of $n$:
\[
    \SW(\cI) = \Theta(\log(n)) \cdot \OPT(\cI).
\] 
Moreover, $\OPT(\cI) = \EQ(\cI)$.
\end{proposition}
\begin{proof}
    Let $\cI$ be the same instance as in Lemma \ref{lem:lb-log}. That is, let $q=q_i \in (0,1)$ be equal for all $i\in [n]$, and $c_i = \left(1-\frac{1}{i}\right) \cdot q + \frac{c}{i}~\forall i \in [n]$.
    From Proposition \ref{prop:equal_scaled_Q} for $S=[n]$, we get 
    \[
    \SW = \sum_{i\in [n]} (q_i-c_i) \leq \left(\sum_{i \in [n]} \frac{1}{i}\right) \cdot \EQ = O(\log(n))\cdot  \EQ.
    \]
    Thus by Lemma \ref{lem:lb-log}, the result follows.
\end{proof}

\subsection{Equal Costs}\label{sec:equal_cost}

In this part, we mirror the results of the earlier section (equal success probabilities) for equal costs, under the limited liability constraint:

\begin{restatable}{lemma}{lemlastpremlowerbound}\label{lem:last_prem_lower_bound}
    Let $c_i=c$ $\forall i\in [n]$. Then, under limited liability, the worst case over all $(q_i)_{i \in [n]}$ has the following asymptotic behavior in function of $n$:
    \[
    \SW = \Theta(\log(n)) \cdot \EQ.\qedhere
    \]
\end{restatable}

\begin{proof}
Let $f(i)\geq 0$ be defined for $i\in [n]$ and be decreasing in $i$. Then let $q_i = c \cdot f(i)$. We get
\[
    \SW = \sum_i q_i - nc = c \cdot \left(\sum_i f(i) - n\right).
\]
By Proposition~\ref{prop:equal}, since costs are all equal to $c$, we have that
\[
    \EQ = \max_{k} \left(1 - \frac{c}{q_k}\right) \cdot \sum_{i \leq k} q_i = \max_{k} \left(1 - \frac{1}{f(k)}\right) \cdot c \cdot \sum_{i \leq k} f(i).
\]
In order to find a large ratio $\SW/\EQ$, we choose to set $f(i)=1+\frac{1}{i}$. We obtain
\[
    \frac{\SW}{\EQ} = \frac{\sum_{i \leq n} f(i) - n}{\max_{k} \left(1 - \frac{1}{f(k)}\right) \cdot \sum_{i \leq k} f(i)} = \frac{\log n}{\max_{k} \frac{k+\log k}{k+1}} \geq \frac{\log n}{2}.
\]
The inequality follows from the fact that $\frac{k+\log k}{k+1} \leq 2$ for all $k\geq 1$.
Thus, we have shown that there exists an instance $\cI$ such that 
\[\SW(\cI) \geq \Omega(\log n) \cdot \EQ(\cI).\]
As for the upper bound, from Proposition \ref{prop:equal_scaled_C} for $S=[n]$, we get 
\[
\SW = \sum_{i\in [n]} (q_i-c_i) \leq \left(\sum_{i \in [n]} \frac{1}{i}\right) \cdot \EQ = O(\log(n))\cdot  \EQ
\]
We conclude that $\SW = \Theta(\log(n))\cdot \EQ.$
\end{proof}

\end{document}